%% file: SAND_submission.tex
\documentclass[a4paper,UKenglish,cleveref,thm-restate,pdfa]{lipics-v2021}
\pdfoutput=1
\hideLIPIcs

\title{Undecidability of the Emptiness Problem for Weak Models of Distributed Computing}

\titlerunning{Undecidability of the Emptiness Problem for Distributed Automata}

\author{Flavio T. {Principato}}{Technical University of Munich, Germany}{flavio.principato@tum.de}{https://orcid.org/0009-0003-2061-5416}{}
\author{Javier {Esparza}}{Technical University of Munich, Germany}{esparza@cit.tum.de}{https://orcid.org/0000-0001-9862-4919}{}
\author{Philipp {Czerner}}{Technical University of Munich, Germany}{czerner@cit.tum.de}{https://orcid.org/0000-0002-1786-9592}{}

\authorrunning{F.\,T. Principato, J. Esparza, and P. Czerner}

\Copyright{Flavio T. Principato, Javier Esparza, and Philipp Czerner}



\nolinenumbers

\EventEditors{Kitty Meeks and Christian Scheideler}
\EventNoEds{2}
\EventLongTitle{4th Symposium on Algorithmic Foundations of Dynamic Networks (SAND 2025)}
\EventShortTitle{SAND 2025}
\EventAcronym{SAND}
\EventYear{2025}
\EventDate{June 9--11, 2025}
\EventLocation{Liverpool, GB}
\EventLogo{}
\SeriesVolume{330}
\ArticleNo{5}

\usepackage{mathtools,accents}

\usepackage{caption,tikz}
\usetikzlibrary{arrows,positioning,backgrounds,math}

\DeclareMathOperator{\dist}{dist}

\newcommand*{\checked}{\rlap{\checkmark}\square}

\newcommand*{\mapsdown}{\hspace{-0.25ex}\rotatebox{-90}{$\mapsto$}}
\newcommand*{\rsnowball}{\rlap{$\rightarrow$}\bigcirc}
\newcommand*{\rface}{\rlap{$\rightarrow$}\phantom{\bigcirc}}
\newcommand*{\lsnowball}{\rlap{$\leftarrow$}\bigcirc}
\newcommand*{\lface}{\rlap{$\leftarrow$}\phantom{\bigcirc}}

\newcommand{\labelledtag}[1]{\tag*{(#1)}\label{#1}}

\theoremstyle{definition}

\newtheorem{definition2}[theorem]{Definition}
\newtheorem{construction}[theorem]{Construction}
\crefname{construction}{Construction}{Constructions}

\begin{document}

\maketitle

\begin{abstract}
	\input{chapters/abstract.tex}
\end{abstract}

\section{Introduction}\label{sec:intro}
\input{chapters/intro.tex}

\section{Preliminaries}\label{sec:prelim}
\input{chapters/prelim.tex}

\section{Simulating Turing Machines with Distributed Automata}\label{sec:groundwork}
\input{chapters/groundwork.tex}

\section{Snowball Fight!}\label{sec:snowball}
\input{chapters/snowballfight.tex}

\section{Conclusion}\label{sec:conclusion}
\input{chapters/conclusion.tex}

\bibliographystyle{plainurl}
\bibliography{references}

\appendix

\section{Appendix to \Cref{sec:groundwork}}
\input{chapters/appendix_groundwork.tex}

\section{Appendix to \Cref{sec:snowball}}
\input{chapters/appendix_snowballfight.tex}

\end{document}

%% file: chapters/abstract.tex
Esparza and Reiter have recently conducted a systematic comparative study of weak asynchronous models of distributed computing, in which a network of identical finite-state machines acts cooperatively to decide properties of the network's graph.
They introduced a distributed automata framework encompassing many different models, and proved that w.r.t.\ their expressive power (the graph properties they can decide) distributed automata collapse into seven equivalence classes.
In this contribution, we turn our attention to the formal verification problem: Given a distributed automaton, does it decide a given graph property?
We consider a fundamental instance of this question -- the \emph{emptiness problem}: Given a distributed automaton, does it accept any graph at all?
Our main result is negative: the emptiness problem is undecidable for six of the seven equivalence classes, and trivially decidable for the remaining class.

%% file: chapters/intro.tex
The concepts of distributed computing can be used to model interactions between a variety of natural or artificial agents, like molecules, cells, microorganisms, or nanorobots.
Typical features of these models are that agents do not have identities, and each agent has very limited computational power and restricted communication capabilities.
Weak models of distributed computing are hence the appropriate paradigm in this context -- in contrast to traditional distributed computing models used to study computer networks.
Examples of such models include population protocols \cite{pp_1,pp_2}, chemical reaction networks \cite{crn}, networked finite-state machines \cite{nfsm}, the weak models of distributed computing of \cite{wmdc}, and the beeping model \cite{beep_1,beep_2}.
Many of these and other models are discussed in comparative surveys \cite{survey_1,survey_2}.

All these weak models share the following features \cite{nfsm}: the communication network can have arbitrary topology; all agents run the same protocol; each agent has a finite number of states, independent of the network size or topology; state transitions only depend on the agent's state and the states of a bounded number of neighbours; nodes do not know their neighbours, in the sense of \cite{ni}.
Nonetheless, they differ in several other aspects.
Esparza and Reiter classified them in \cite{classification} according to four parameters:

\begin{itemize}
	\item\emph{Detection.}
		In some models, agents can only detect the \emph{existence} of neighbours in a given state ($\mathtt{d}$), while in others they can \emph{count} their number up to a fixed threshold ($\mathtt{D}$).
	\item\emph{Acceptance.}
		A given input is accepted or rejected by all agents reaching a respective state.
		Some models require agents to \emph{halt} once they reach an accepting or rejecting state ($\mathtt{a}$), while others only require \emph{stable consensus} ($\mathtt{A}$), i.e., agents can keep changing their decision, as long as they eventually agree on acceptance or rejection.
	\item\emph{Selection.}
		In each step of the execution of a model, a certain selection of agents acts.
		In \emph{synchronous} models (\texttt{\$}), all agents are selected in every step.
		Models with \emph{liberal} selection ($\mathtt{s}$) select an arbitrary subset of agents in each step.
		Finally, models with \emph{exclusive} selection ($\mathtt{S}$) select exactly one agent at a time.
	\item\emph{Fairness.}
		Different assumptions about the occurrence of the permitted selections can be made.
		Some models only ensure that every agent will be selected infinitely many times.
		We call this \emph{weak fairness} ($\mathtt{f}$).
		Others use stochastic-like selection, which guarantees that any finite sequence of permitted selections will occur infinitely often.
		We call this \emph{strong fairness} ($\mathtt{F}$).
\end{itemize}

In \cite{classification}, Esparza and Reiter studied the expressive power of all possible combinations of these parameters.
For this, they introduced a generic model, called \emph{distributed automata}, which can be equipped with any parameter combination -- for example, one can study the class of $\mathtt{dA\texttt{\$}f}$-distributed automata.
The behaviour of a distributed automaton is described by a finite-state machine, inputs are labelled graphs and the output is boolean (acceptance/rejection).
Intuitively, each node of the input graph becomes an agent running a copy of the finite-state machine where the initial state depends on the node's label, and its transitions depend on the agent's own and the agent's neighbours' states.
At each step of the execution a subset of the nodes is selected by a scheduler, and each selected agent executes a transition of the finite-state machine, thereby moving into a new state.
Note that this adds a component of nondeterminism, such that multiple runs on the same input graph could a priori yield different results.
This ambiguity is resolved by requiring the \emph{consistency condition}, which essentially states that for a given input graph the distributed automaton's output must always be the same.
The set of graphs accepted by a given distributed automaton forms a graph language; we say the automaton decides the graph language.
A model's expressive power is determined by the class of graph languages it can decide.

The main result of \cite{classification} is that all the possible combinations of parameters collapse into seven equivalence classes w.r.t.\ their expressive power.
Notably, all proofs of equivalence are constructive.
In particular, every class of distributed automata was shown to be equivalent to some class with liberal selection ($\mathtt{s}$).
For this reason, we can often omit the selection parameter when referring to an equivalence class $xyz \in \{\mathtt{d},\mathtt{D}\}\{\mathtt{a},\mathtt{A}\}\{\mathtt{f},\mathtt{F}\}$, implicitly meaning $xy\mathtt{s}z$.
Furthermore, the classes $\mathtt{da}z$ with $z \in \{\mathtt{f},\mathtt{F}\}$ were shown to both have trivial expressive power, i.e., every $\mathtt{da}z$-distributed automaton either accepts all labelled graphs or none.
The seven distinct equivalence classes assemble in a hierarchy of expressive power shown in \Cref{fig:diagram}.
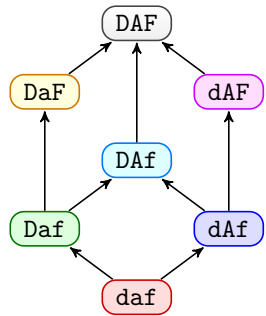
\begin{figure}[h]
	\centering
	\input{figs/diagram.tex}%
	\caption{The hierarchy of expressive power of classes of distributed automata as proven in \cite{classification} --- The arrows indicate strictly increasing expressive power.}
	\label{fig:diagram}
\end{figure}
Finally, we mention that Czerner \emph{et al.}\ \cite{decision_power} characterized which labelling properties each class can decide, i.e., properties depending only on the labelling and not on the structure of the graph.

Distributed automata are designed to decide properties of graphs, e.g., whether a given graph is cyclic or whether at least one of its nodes is labelled as red.
So, we are interested in the verification problem: Does the graph language $L(A)$ of a given distributed automaton $A$ satisfy a given property, i.e., does $L(A) \subseteq \mathcal{C}$ hold, where $\mathcal{C}$ is the class of graphs with the desired property?
This is equivalent to asking if  $L(A) \cap \overline{\mathcal{C}} = \emptyset$ holds, where $\overline{\mathcal{C}}$ denotes the complement of $\mathcal{C}$.
Since distributed automata are closed under intersection (as we show in this paper), verification reduces to the \emph{emptiness problem} whenever $\overline{\mathcal{C}}$ is decidable by distributed automata.

We show that the emptiness problem is undecidable for six of the seven classes, and trivially decidable for the class of $\mathtt{daf}$-distributed automata%
\footnote{Since a $\mathtt{daf}$-distributed automaton either accepts all graphs or no graph, emptiness can be decided by checking whether the automaton accepts, e.g., the graph with one node and no edges, which is easy.}.
We reduce from the halting problem for Turing machines on blank tape.
For each class $xyz$, we define a suitable family of labelled graphs to represent a finite tape section.
Then we construct an $xyz$-distributed automaton that solves two tasks: it decides whether the graph belongs to the suitable family and, if so, simulates the execution of a given Turing machine until either the machine halts, in which case the automaton accepts, or the head exceeds the finite tape section, in which case the automaton rejects.
So, we need families of labelled graphs that can be used to represent arbitrarily long finite tape sections, and are decidable by $xyz$-distributed automata.
This is especially difficult for the classes of the form $\mathtt{Da}z$, which have extremely limited expressive power.
Identifying a suitable family of labelled graphs for this case is the main contribution of our paper.

\smallskip\noindent\textit{Related work.}
Kuusisto and Reiter have studied the emptiness problem for a specific class of automata on directed graphs working synchronously \cite{KuusistoR20}.
This work predates \cite{classification}, which considers a large variety of models on undirected graphs%
\footnote{Unfortunately, \cite{KuusistoR20} and \cite{classification} use the name distributed automata with different meanings: \cite{KuusistoR20} for a specific automata model, and \cite{classification} as a generic name for all the automata classes in their classification.}.

\smallskip\noindent\textit{Structure of the paper.}
\Cref{sec:prelim} contains preliminaries.
\Cref{sec:groundwork} identifies suitable families of graphs for the more capable classes $\mathtt{DA}z$ and $\mathtt{dA}z$, and prepares the ground for \Cref{sec:snowball}, which introduces a suitable family for $\mathtt{Da}z$.
\Cref{sec:conclusion} presents conclusions.
Formal descriptions of the constructions and detailed proofs are given in the appendices.

%% file: figs/diagram.tex
\definecolor{draw_daf}{HTML}{cc0000}
\definecolor{draw_Daf}{HTML}{007700}
\definecolor{draw_dAf}{HTML}{0000ff}
\definecolor{draw_DAf}{HTML}{0077ff}
\definecolor{draw_DaF}{HTML}{cc7700}
\definecolor{draw_dAF}{HTML}{cc00ff}
\definecolor{draw_DAF}{HTML}{444444}

\definecolor{fill_daf}{HTML}{ffdddd}
\definecolor{fill_Daf}{HTML}{ddffdd}
\definecolor{fill_dAf}{HTML}{ddddff}
\definecolor{fill_DAf}{HTML}{ddffff}
\definecolor{fill_DaF}{HTML}{ffffdd}
\definecolor{fill_dAF}{HTML}{ffddff}
\definecolor{top_DAF} {HTML}{ffffff}
\definecolor{bot_DAF} {HTML}{eeeeee}

\begin{tikzpicture}[semithick,
	on grid, every node/.style={rounded corners=1.2ex, inner sep=0ex, minimum height=3ex, minimum width=6ex},
	>=stealth', shorten >=0.5pt]
	\def\dv{6ex}
	\def\dh{8ex}

	\node [draw=draw_daf,fill=fill_daf]									(daf)                                  {$\mathtt{daf}$};
	\node [draw=draw_Daf,fill=fill_Daf]									(Daf) [above left=\dv and \dh of daf]  {$\mathtt{Daf}$};
	\node [draw=draw_dAf,fill=fill_dAf]									(dAf) [above right=\dv and \dh of daf] {$\mathtt{dAf}$};
	\node [draw=draw_DAf,fill=fill_DAf]									(DAf) [above=2*\dv of daf]             {$\mathtt{DAf}$};
	\node [draw=draw_DaF,fill=fill_DaF]									(DaF) [above=2*\dv of Daf]             {$\mathtt{DaF}$};
	\node [draw=draw_dAF,fill=fill_dAF]									(dAF) [above=2*\dv of dAf]             {$\mathtt{dAF}$};
	\node [draw=draw_DAF, shade,top color=top_DAF,bottom color=bot_DAF]	(DAF) [above=2*\dv of DAf]             {$\mathtt{DAF}$};

	\draw[->]
		(daf) edge (Daf)
		(daf) edge (dAf)
		(Daf) edge (DAf)
		(Daf) edge (DaF)
		(dAf) edge (DAf)
		(dAf) edge (dAF)
		(DAf) edge (DAF)
		(DaF) edge (DAF)
		(dAF) edge (DAF)
		;
\end{tikzpicture}

%% file: chapters/prelim.tex
For sets $A$ and $B$, $\mathcal{P}(A)$ denotes the powerset of $A$ and $B^A$ denotes the set of all functions from $A$ to $B$.
Given a partial function $f \colon {A \to B}$, we write $f(a) = \bot$ to denote that $f$ is not defined for $a \in A$.
We use the convention $0 \in \mathbb{N}$ and define $[n] \coloneq \{0,...,n\}$ for $n \in \mathbb{N}$ and $[-1] \coloneq \emptyset$.

\smallskip\noindent\emph{Words.} Given a word $w \in \Sigma^*$ over an alphabet $\Sigma$, we let $|w|$ denote the length of $w$  and $w_{i...j}$ the substring of $w$ from the $i^\text{th}$ to the $j^\text{th}$ symbol for $i,j \in [|w|-1]$,  with $w_{i...j} \coloneq \varepsilon$ for $j<i$ and $w_i \coloneq w_{i...i}$.
Note that we use $0$-indexing.
For $a \in \Sigma$, $a^\omega$ denotes the infinite repetition of $a$.

\smallskip\noindent\emph{Graphs.} We implicitly assume graphs to be non-empty, finite, undirected, unweighted, and connected.
A labelled graph over a set of labels $\Lambda$ is a triple $G=(V,E,\lambda)$, where $V$ and $E$ are sets of nodes and edges, and $\lambda \colon V \to \Lambda$ is a labelling function that assigns a label to each node $v \in V$.
For $U \subseteq V, v \in V$, the function $\dist(U,v)$ denotes the length of the shortest path from any node in $U$ to $v$.

\smallskip\noindent\emph{Turing machines.} A Turing machine (TM) is a $7$-tuple $T = (Q,q_0,F,\Gamma,\Sigma,\square,\delta)$ with a finite set of states  $Q$, an initial state $q_0 \in Q$, a set of accepting states $F \subseteq Q$, a tape alphabet $\Gamma$, an input alphabet $\Sigma \subseteq \Gamma$, the blank symbol $\square \in \Gamma \setminus \Sigma$, and a (partial) transition function $\delta \colon {Q \times \Gamma \to Q \times \Gamma \times \{-1,+1\}}$ with $\delta(f,\gamma) = \bot$ for all $f \in F, \gamma \in \Gamma$.
In our model, the TM tape is only unbounded to the right.
The TM head starts on the leftmost cell and moves right ($+1$) or left ($-1$) in each transition.

\subsection{Distributed Machines}

A \emph{distributed machine} operating on a \emph{labelling alphabet} $\Lambda$ with counting bound $\beta \in \mathbb{N} \setminus \{0\}$ is a $5$-tuple $M = (Q,\delta_0,\delta,Y,N)$ with a finite set of \emph{states} $Q$, an \emph{initialization function} $\delta_0 \colon {\Lambda \to Q}$, a \emph{transition function} $\delta \colon {Q \times {[\beta]}^Q \to Q}$, and disjoint sets $Y,N \subseteq Q$ of \emph{accepting} and \emph{rejecting states}, respectively.

A distributed machine $M$ runs on a labelled graph $G = (V, E, \lambda)$.
Intuitively, at each node $v \in V$ an \emph{agent} starts in the state $\delta_0(\lambda(v))$.
The machine proceeds in steps, where in each step a set $S$ of nodes is selected.
Each agent $v \in S$ transitions to a new state according to the function $\delta$, which depends on $v$'s and $v$'s neighbours' current states.

We proceed to formalize this intuition.
A \emph{configuration} $C$ of $M$ on $G$ is a function $C \colon {V \to Q}$ that assigns a \emph{current state} to every agent of $G$.
For every configuration $C$ and agent $v$, we define the function $\mathcal{N}_v^C \colon {Q \to [\beta], q \mapsto \min\{n_q(v),\beta\}}$, where $n_q(v) \in \mathbb{N}$ is the number of agents $w$ neighbouring $v$ with $C(w)=q$.
To simplify notation, we define for $Q = R \times S \times T$ that $\mathcal{N}_v^C(r,*,t) = \sum_{s \in S} {\mathcal{N}_v^C(r,s,t)}$.
Given two configurations $C,D$ and a \emph{selection} $S \subseteq V$, we let $C \stackrel{S}{\longrightarrow}_M D$ denote that $D(v) = \delta\big(C(v),\mathcal{N}_v^C\big)$ for all $v \in S$ and $D(v) = C(v)$ otherwise, that is, the machine transitions from $C$ to $D$ by performing state transitions according to $\delta$ for exactly the agents selected by $S$.
We write $C \to_M D$ if there exists some selection $S \subseteq V$ such that $C \stackrel{S}{\longrightarrow}_M D$.

\subsection{Distributed Automata}

In principle, a distributed automaton consists of a distributed machine and a scheduler specifying which sequences of selections are allowed.
All schedulers guarantee the minimal assumption that every node is selected infinitely often.

A \emph{schedule} is a sequence of selections $\sigma = {(S_i)}_{i \in \mathbb{N}}$ such that for every $v \in V$ there exist infinitely many $i \in \mathbb{N}$ with $v \in S_i$, i.e., every agent is selected infinitely often.
A sequence of configurations $\rho = {(C_i)}_{i \in \mathbb{N}}$ with $C_0 = \delta_0 \circ \lambda$ and $C_i \to_M C_{i+1}$ for all $i \in \mathbb{N}$ is called a \emph{run} of $M$ on $G$.
We say that $\rho$ is \emph{induced} by $\sigma$ if{}f $C_i \stackrel{S_{i}}{\longrightarrow}_M C_{i+1}$ for all $i \in \mathbb{N}$.
A configuration $C$ is called \emph{accepting} if{}f $C(V) \subseteq Y$ and \emph{rejecting} if{}f  $C(V) \subseteq N$.
A run $\rho = {(C_i)}_{i \in \mathbb{N}}$ is called \emph{accepting} (\emph{rejecting}) if{}f there exists an index $I \in \mathbb{N}$ such that $C_i$ is accepting (rejecting) for all $i \geq I$.

A \emph{scheduler} is a pair $\Sigma = (s,f)$.
The \emph{selection constraint} $s$ maps any graph $G = (V, E, \lambda)$ to a subset of $\mathcal{P}(V)$ of \emph{permitted} selections, such that $\bigcup_{S \in s(G)} {S} = V$, i.e., every agent is in at least one permitted selection.
The \emph{fairness constraint} $f$ maps $G$ to a subset of ${s(G)}^\mathbb{N}$ of \emph{fair} schedules.
For a distributed machine $M$ and a scheduler $\Sigma$, a run of $M$ on a labelled graph $G$ is called \emph{fair} if{}f it is induced by a fair schedule.

A pair $A = (M,\Sigma)$ consisting of a distributed machine $M$ and a scheduler $\Sigma$ satisfies the \emph{consistency condition} on a class of graphs $\mathcal{C}$ if{}f for every given graph $G \in \mathcal{C}$, either all fair runs of $M$ on $G$ are accepting or all are rejecting.
Intuitively, whether $M$ accepts $G \in \mathcal{C}$ does not depend on which fair run we consider.
In this case, we call $A$ a \emph{$\mathcal{C}$-distributed automaton} and say that \emph{$A$ accepts $G$} or \emph{$A$ rejects $G$}, respectively.
$L(A) \subseteq \mathcal{C}$ denotes the language of $\Lambda$-labelled graphs accepted by $A$.
A \emph{distributed automaton} is a $\mathcal{G}$-distributed automaton, where $\mathcal{G}$ denotes the set of all graphs.

\subsection{Classes of Distributed Automata}

We formalize the parameters \emph{detection}, \emph{acceptance}, \emph{selection} and \emph{fairness} mentioned in the introduction.
The first two concern the distributed machine $M = (Q,\delta_0,\delta,Y,N)$, while the other two concern the scheduler $\Sigma = (s,f)$.
Let $G = (V,E,\lambda)$ be any $\Lambda$-labelled graph.

\smallskip\noindent\textbf{Detection.}
A distributed machine with counting bound $\beta = 1$ has \emph{existence detection} ($\mathtt{d}$), while a machine with $\beta > 1$ has \emph{counting detection} ($\mathtt{D}$).

\smallskip\noindent\textbf{Acceptance.}
A distributed machine with \emph{halting acceptance} ($\mathtt{a}$) fulfils $\delta(q,\mathcal{N}) = q$ for all $q \in Y \cup N, \mathcal{N} \in [\beta]^Q$.
Otherwise, it accepts by \emph{stable consensus} ($\mathtt{A}$).

\smallskip\noindent\textbf{Selection.}
A \emph{synchronous} (\texttt{\$}) scheduler fulfils $s(G) = \{V\}$, i.e., the only permitted selection is all of $V$.
The other extreme, a \emph{liberal} ($\mathtt{s}$) scheduler, is characterized by $s(G) = \mathcal{P}(V)$, i.e., any selection is permitted.
Lastly, the scheduler is \emph{exclusive} ($\mathtt{S}$) if{}f $s(G) = \{\{v\}\}_{v \in V}$, i.e., in every step, exactly one agent is selected.

\smallskip\noindent\textbf{Fairness.}
For a \emph{weakly fair} ($\mathtt{f}$) scheduler, $f(G)$ contains all schedules in ${s(G)}^\mathbb{N}$.
(Recall that a schedule, by definition, selects every agent infinitely often.)
A schedule ${(S_i)}_{i \in \mathbb{N}}$ is strongly fair if{}f for every finite sequence of permitted selections ${(T_i)}_{i \in [n]} \in {s(G)}^*$, there are infinitely many $j \in \mathbb{N}$, such that ${(S_{j+i})}_{i \in [n]} = {(T_i)}_{i \in [n]}$.
For a \emph{strongly fair} ($\mathtt{F}$) scheduler, $f(G) \subseteq {s(G)}^\mathbb{N}$ contains exactly the strongly fair schedules.

The expressive power of all possible combinations of these parameters is analysed in \cite{classification}.
As mentioned in the introduction, after merging classes with the same expressive power one obtains the hierarchy shown in \Cref{fig:diagram}.
For the purpose of studying the emptiness problem, we are interested in the six classes with non-trivial expressive power, as the problem is trivially decidable for the $\mathtt{daf}$ class.

Recall that $xyz$ implicitly means $xy\mathtt{s}z$.
However, \cite{classification} also proves the following result:

\begin{proposition}[\text{\cite[Lemma 3 (5) and Theorem 5]{classification}}]\label{prop:sync}
	For all $xy \in \{\mathtt{d},\mathtt{D}\}\{\mathtt{a},\mathtt{A}\}$, the classes $xy\mathtt{f}$ and \emph{$xy\texttt{\$}\mathtt{f}$} have the same expressive power.
	Moreover, given a distributed automaton in one of the classes, one can effectively construct an equivalent automaton in the other class.
\end{proposition}

\noindent Therefore, in order to prove the undecidability of the emptiness problem for  $\mathtt{DAf}$, $\mathtt{dAf}$, and $\mathtt{Daf}$, it suffices to prove it for $\mathtt{DA}\texttt{\$}\mathtt{f}$, $\mathtt{dA}\texttt{\$}\mathtt{f}$, and $\mathtt{Da}\texttt{\$}\mathtt{f}$.
We will make use of this.

%% file: chapters/groundwork.tex
The emptiness problem for $xyz$-distributed automata, where $xyz \in \{\mathtt{d},\mathtt{D}\}\{\mathtt{a},\mathtt{A}\}\{\mathtt{f},\mathtt{F}\}$, consists of deciding whether a given $xyz$-distributed automaton $A$ satisfies $L(A) = \emptyset$.
We prove undecidability of the emptiness problem for $xy \neq \mathtt{da}$ by  reduction from the halting problem for Turing machines on blank tape.
Formally, we show:

\begin{restatable}[Undecidability of the Emptiness Problem for Distributed Automata]{theorem}{reduction}\label{thm:reduction}
	Let $xyz \in \{\mathtt{d},\mathtt{D}\}\{\mathtt{a},\mathtt{A}\}\{\mathtt{f},\mathtt{F}\}$ with $xy \neq \mathtt{da}$.
	Given a Turing machine $T$, one can effectively construct an $xyz$-distributed automaton $A^T$, such that $L(A^T) \neq \emptyset$ if{}f $T$ halts on blank tape.
	The emptiness problem for $xyz$-distributed automata is thus undecidable.
\end{restatable}

The proof of the theorem uses the same idea for every class $xyz$, which we proceed to describe.
First, we observe that distributed automata are closed under intersection:

\begin{restatable}[Closure under Intersection]{lemma}{intersection}\label{lem:intersection}
	Let $xywz \in \{\mathtt{d},\mathtt{D}\}\{\mathtt{a},\mathtt{A}\}\{\emph{\texttt{\$}},\mathtt{s},\mathtt{S}\}\{\mathtt{f},\mathtt{F}\}$.
	Given two $xywz$-distributed automata $A_1, A_2$ operating on the same labelling alphabet $\Lambda$, we can effectively construct an $xywz$-distributed automaton $A$, such that $L(A) = L(A_1) \cap L(A_2)$.
	Moreover, this result remains valid even when $A_2$ is only an $L(A_1)$-distributed automaton.
\end{restatable}

Note that we should not omit the selection parameter $w \in \{\texttt{\$},\mathtt{s},\mathtt{S}\}$ here, as the results of \cite{classification} (which justified omitting the selection parameter elsewhere) do not necessarily hold for $L(A_1)$-distributed automata.

\begin{claimproof}[Proof sketch. Full proof in \Cref{app:intersection}]
	We outline the proof of the first claim.
	The construction is very similar to the product construction known from DFAs.
	Intuitively, $A$ executes $A_1,A_2$ in parallel, transitioning between states in $Q_1 \times Q_2$ according to $\delta((q_1, q_2), \mathcal{N})) = (\delta_1(q_1, \mathcal{N}), \delta_2(q_2, \mathcal{N}))$.
	Further, $(q_1,q_2) \in Y$ if{}f $q_1 \in Y_1$ and $q_2 \in Y_2$, and $(q_1,q_2) \in N$ if{}f $q_1 \in N_1$ or $q_2 \in N_2$.
\end{claimproof}

Given a Turing machine $T$, the proof of \Cref{thm:reduction} consists of exhibiting a family of labelled graphs $\mathcal{L} = \{G_n\}_{n \geq 1}$ together with a distributed automaton $A^L$ and an $\mathcal{L}$-distributed automaton $A^H$ satisfying
(1)\label{itm:sat_1} $A^L$ is an $xyz$-distributed automaton with $L(A^L) = \mathcal{L}$, and
(2)\label{itm:sat_2} $A^H$, for a graph $G_n \in \mathcal{L}$, accepts if{}f $T$ halts after visiting at most $n$ tape cells, and rejects otherwise; intuitively, a run of $A^H$ on $G_n$ simulates the execution of $T$ as long as the TM head never exceeds the first $n$ tape cells.

Notice that this implies the existence of the distributed automaton $A^T$ of \Cref{thm:reduction}.
For this, let $A^T$ be an automaton deciding the language $L(A^L) \cap L(A^H)$, which exists by \Cref{lem:intersection}.
We show that $L(A^T) \neq \emptyset$ if{}f $T$ halts:
If $L(A^T) \neq \emptyset$, let $G \in L(A^L) \cap L(A^H)$.
By \hyperref[itm:sat_1]{(1)}, $G = G_n$ for some $n \geq 1$.
By \hyperref[itm:sat_2]{(2)}, $T$ halts visiting at most $n$ tape cells.
In particular, $T$ halts.
Conversely, if $T$ halts, then it does so after a finite number of steps, visiting a finite number $n$ of tape cells.
So by \hyperref[itm:sat_1]{(1)} and \hyperref[itm:sat_2]{(2)} both $A^L$ and $A^H$ accept $G_n$, yielding $G_n \in L(A^L) \cap L(A^H)$.
In particular, $L(A^T) \neq \emptyset$.

We will construct $A^L$ and $A^H$ for the classes $xy\texttt{\$}\mathtt{f}$ with $xy \in \{\mathtt{DA},\mathtt{dA},\mathtt{Da}\}$.
By \Cref{prop:sync}, this implies that we can effectively construct an equivalent automaton for each class $xy\mathtt{f}$, and by the hierarchy of expressive power this can be lifted to the respective $xy\mathtt{F}$ class as well.
This is thus sufficient to prove \Cref{thm:reduction} for all classes%
\footnote{In fact, it would even suffice to only construct $A^L$ and $A^H$ for the classes $xy\texttt{\$}\mathtt{f}$ with $xy \in \{\mathtt{dA},\mathtt{Da}\}$, omitting $xy = \mathtt{DA}$.
Including the construction for $xy = \mathtt{DA}$ was a didactic choice, as it lays the groundwork and helps motivate the other two constructions.}.

The rest of the section is structured as follows.
We first consider the class $\mathtt{DA\texttt{\$}f}$, for which Subsection \ref{sec:nlg} defines a suitable family $\mathcal{L}$ and the automaton $A^L$, and Subsection \ref{sec:tm_head} defines the automaton $A^H$.
In Subsection \ref{sec:nqlg}, we find a slightly more complex family and automaton for the $\mathtt{dA\texttt{\$}f}$ class.

\subsection{Representing the Tape}\label{sec:nlg}
\input{chapters/nlg.tex}

\subsection{Simulating the Head}\label{sec:tm_head}
\input{chapters/tm_head.tex}

\subsection{Weakly Representing the Tape}\label{sec:nqlg}
\input{chapters/nqlg.tex}

%% file: chapters/nlg.tex
The simplest possible representation of a finite TM tape section is a linear graph where each node represents a tape cell and the labelling gives rise to a notion of positive/right and negative/left directions.
A labelling achieving this is numbering the nodes modulo $3$, starting with $0$; an agent with numbering $i$ can identify its left and right neighbour as the unique neighbour with numbering $(i-1)\bmod{3}$ and $(i+1)\bmod{3}$, respectively.
We call this family of graphs \textit{numbered linear graphs} (NLGs) and denote it by $\mathsf{NLG}$.
\Cref{fig:nlg} shows an example for an NLG.
The \textit{origin node} of an NLG is the unique node that has numbering $0$ and no left neighbour, i.e., no neighbour numbered $2$.

\begin{figure}[h]
	\centering
	\input{figs/nlg_ex.tex}%
	\caption{A numbered linear graph of length $7$.}
	\label{fig:nlg}
\end{figure}
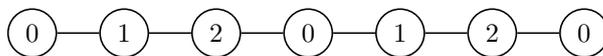

\begin{restatable}[Decidability of Numbered Linear Graphs]{lemma}{lemNlg}\label{lem:nlg}
	We can effectively construct a \emph{$\mathtt{DA\texttt{\$}f}$}-distributed automaton that decides $\mathsf{NLG}$.
\end{restatable}

\begin{claimproof}[Proof sketch. Full proof in \Cref{app:nlg}.]
	We sketch the relevant construction.	
	The labelling alphabet is $\Lambda^L = \mathbb{Z}_3$ and the agents have states in $Q^L = (\Lambda^L \times \{0,1\}) \cup \{\bot\}$.
	The first component of the state is the numbering, which is taken directly from the labelling and stays static, and the second component is the agent's current guess whether the graph has an origin node.
	$\bot$ is an error state that, once it occurs, propagates through the graph, that is, if any neighbour of an agent $v$ is in state $\bot$, $v$ transitions to $\bot$ as well.
	The only accepting states are those with guess $1$; all other states, most notably $\bot$, are rejecting.
	
	If an agent $v$ has the same numbering as one of its neighbours, it transitions to the error state $\bot$ immediately; the same happens if $v$ detects two of its neighbours having the same numbering as each other.
	A graph with faulty numbering or nodes of degree greater than $2$ will therefore always produce an error.
	
	Lastly, we have to distinguish numbered linear graphs from circular graphs with correct numbering.
	This is where the guess component comes in: If an origin node exists, it is the first to change its guess to $1$, which then propagates through the graph to accept.
	If no origin node exists, no agent will be the first to guess $1$ and thus all agents will stay in rejecting states indefinitely.
\end{claimproof}

%% file: figs/nlg_ex.tex
\begin{tikzpicture}[semithick,
	on grid, every node/.style={circle, draw, minimum width=4.2ex}]
	\def\d{8ex}
	
	\pgfmathsetmacro{\num}{7}
	
	\node (0) {$0$};
	\foreach \i in {1,...,\the\numexpr\num-1\relax} {
		\pgfmathtruncatemacro{\label}{mod(\i, 3)};
		\node (\i) [right=\d of \the\numexpr\i-1\relax] {$\label$};
	}
	
	\foreach \i in {1,...,\the\numexpr\num-1\relax} {
		\draw (\the\numexpr\i-1\relax) -- (\i);
	}
\end{tikzpicture}

%% file: chapters/tm_head.tex
Given an NLG of length $n \geq 1$ to represent a finite TM tape section, where each agent represents a tape cell, we can simulate the behaviour of the TM head by augmenting the states.
Every agent carries a symbol of the tape alphabet in an additional component.
Further, exactly one agent indicates that the TM head is currently on its corresponding tape cell and saves the TM state.
When the TM head performs a transition, it changes the agent's tape symbol and the TM state and indicates in which direction it will move next.
The agent in the respective direction detects this and takes over as TM head.
This clearly mimics the behaviour of an actual Turing machine on a tape of finite length $n$.

Recall that in our model, the TM tape is only unbounded to the right and the TM head starts on the leftmost cell.
So, the simulation starts with all agents carrying a blank symbol and the TM head on the unique origin node's tape cell.
The TM head state will move through the NLG until the TM halts, producing an accepting state, or the TM runs out of space in the finite TM tape section, producing a rejecting state.
The accepting (rejecting) state then propagates through the graph to accept (reject).

For any given Turing machine $T$, we can construct an equivalent Turing machine $T_\infty$ that either halts or visits every tape cell; in particular, it can never happen that the TM gets stuck in a loop.
By simulating $T_\infty$ rather than $T$ itself, we can be sure that the simulation always either accepts by halting or rejects by running out of space.

Carrying all of this out formally in \Cref{app:tm_head} yields the following result.

\begin{restatable}[Simulating the Head]{lemma}{simulation}\label{cor:always_acc/rej}
	Given a Turing machine $T$, we can effectively construct an $\mathsf{NLG}$-distributed automaton $A^H$ in the class \emph{$\mathtt{da\texttt{\$}f}$} such that:
	\begin{itemize}
		\item If $T$ does not halt on blank tape, then $A^H$ rejects all numbered linear graphs.
		\item If $T$ halts on blank tape, then there exists a threshold $n_0 \in \mathbb{N}$, such that $A^H$ accepts all numbered linear graphs of length $n \geq n_0$ and rejects all numbered linear graphs of length $n < n_0$.
	\end{itemize}%
	
\end{restatable}

The threshold $n_0$ will turn out to be the finite number of tape cells that $T_\infty$ visits before it halts.
We construct $A^H$ in the trivial class $\mathtt{da\texttt{\$}f}$, as this enables us to lift the construction to any of the non-trivial classes of distributed automata by the hierarchy of expressive power%
\footnote{Note that \Cref{cor:always_acc/rej} does not constitute a contradiction to $\mathtt{da\texttt{\$}f}$-distributed automata having trivial expressive power, as $A^H$ is only an $\mathsf{NLG}$-distributed automaton, not a distributed automaton.}.
By doing so for the class $\mathtt{DA\texttt{\$}f}$, we now have all the puzzle pieces to execute the proof of \Cref{thm:reduction} for $\mathtt{DA\texttt{\$}f}$-distributed automata, as outlined at the beginning of \Cref{sec:groundwork}.
The details can be found in \Cref{app:proof_1}.

%% file: chapters/nqlg.tex
It is easy to see that $\mathtt{dA\texttt{\$}f}$-automata cannot decide $\mathsf{NLG}$.
A formal proof can be given by adapting \cite[Proposition D.1]{decision_power}, but we sketch the argument.

Take an NLG of length at least $3$ and duplicate its origin node $v_0$ into two nodes $v_{01}$ and $v_{02}$, which become identical left neighbours of the second node $v_1$.
The new graph is not a linear graph and should be rejected.
Let $A$ be a $\mathtt{dA\texttt{\$}f}$-distributed automaton.
The agents at $v_{01}$ and $v_{02}$ start in the same state and both have $v_1$ as their only neighbour.
As $A$ selects all agents in every step (\texttt{\$}), $v_{01}$ and $v_{02}$ will always have the same state as each other.
Since agents of $A$ cannot count ($\mathtt{d}$), $v_1$ will thus never be able to detect that it has two left neighbours.

We exhibit a new family of suitable graphs for $\mathtt{dA\texttt{\$}f}$-automata, namely \emph{numbered quasi-linear graphs} (NQLGs).
Intuitively, an NQLG is obtained similarly to the counterexample above: take an NLG and replicate each node an arbitrary number of times where every replica of a node has at least one predecessor (successor) if{}f the original node had a predecessor (successor).
\Cref{fig:nqlg_ex} shows an example of an NQLG with two replicas of the first node, three of the second node, etc.
\begin{figure}[h]
	\centering
	\input{figs/nqlg_ex.tex}%
	\caption{A numbered quasi-linear graph of length $5$.}
	\label{fig:nqlg_ex}
\end{figure}
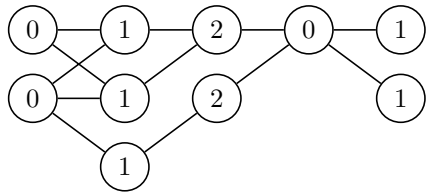
The formal definition takes a bit of a different angle on the same family of graphs:

\begin{definition2}[Numbered Quasi-Linear Graphs]\label{def:nqlg}
	Let $G = (V,E,\lambda)$ with $\lambda \colon {V \to \mathbb{Z}_3}$.
	Further, let $O = \{v \in V \mid \lambda(v) = 0,\,\lambda(w) \neq 2 \text{ for all neighbours $w$ of $v$}\}$.
	We call $G$ a \emph{numbered quasi-linear graph} of length $n \in \mathbb{N} \setminus \{0\}$ if{}f $O$ is non-empty and
	\begin{description}
		\item[(QL1)] $\forall v \in V \colon {\lambda(v) = \dist(O,v)\bmod{3}}$,\label{itm:nqlg_numbering}
		\item[(QL2)] $\forall \{v,w\} \in E \colon {\lambda(v) \neq \lambda(w)}$, and\label{itm:nqlg_neighbours}
		\item[(QL3)] for all $v \in V$\label{itm:nqlg_length}: $\big(\exists \{v,w\} \in E \colon {\dist(O,w) = \dist(O,v)+1}\big) \iff \dist(O,v) < n-1.$
	\end{description}
	The set $O$ is called the \emph{origin set}; the elements of $O$ are called \emph{origin nodes}.
	We denote the family of numbered quasi-linear graphs by $\mathsf{NQLG}$.
\end{definition2}

Observe that every NLG is an NQLG, i.e., $\mathsf{NLG} \subseteq \mathsf{NQLG}$.
Notice also that \hyperref[itm:nqlg_length]{(QL3)} requires every node with a distance to the origin strictly smaller than $n-1$ to have a successor and prohibits nodes with distance $n-1$ to do so.
This enables us to unambiguously assign a length to an NQLG.

In the rest of the section, we show that
(1)\label{itm:nqlg_1} $\mathtt{dA\texttt{\$}f}$-distributed automata can decide $\mathsf{NQLG}$ (the counterpart to \Cref{lem:nlg}), and
(2)\label{itm:nqlg_2} the TM head simulation $A^H$ of \Cref{cor:always_acc/rej} works on NQLGs as well.
We start with part \hyperref[itm:nqlg_1]{(1)}:

\begin{restatable}[Decidability of Numbered Quasi-Linear Graphs]{lemma}{lemNqlg}\label{lem:nqlg}
	We can effectively construct a \emph{$\mathtt{dA\texttt{\$}f}$}-distributed automaton that decides $\mathsf{NQLG}$.
\end{restatable}

\begin{claimproof}[Proof sketch. Full proof in \Cref{app:nqlg}.]
	We sketch the necessary construction.	
	It builds upon the one in \Cref{lem:nlg}, using the same labelling alphabet $\Lambda^L$ and states in $Q^L = (\Lambda^L \times \{0,1,2\}) \cup \{\bot\}$.
	The first component of the state and the $\bot$ state work exactly as in \Cref{lem:nlg}.
	The second component indicates one of three stages: $0$ -- the initial stage, $1$ -- origin set detected, or $2$ -- origin set and end of graph detected.
	The only accepting states are those in stage $2$, all other states are rejecting.
	
	If the given graph is an NQLG, first, all agents in the origin set $O$ transition to stage $1$ at the same time.
	In the $d^\text{th}$ subsequent step, all agents at distance $d$ from the origin set $O$ transition to stage $1$ simultaneously, until finally the agents at distance $n-1$ reach stage $1$, which are exactly the agents that do not have a successor.
	These agents at the end of the graph then continue to stage $2$.
	Again, in the $d^\text{th}$ step after that, all nodes at distance $n-1-d$ reach stage $2$ simultaneously, ultimately accepting the graph.
	
	If, however, the graph is not an NQLG, the automaton can always detect this.
	As for circular graphs in \Cref{lem:nlg}, if $O$ is empty, all agents will passively reject by never leaving stage $0$.
	If \hyperref[itm:nqlg_neighbours]{(QL2)} is violated, this will be detected and rejected immediately.
	A violation of \hyperref[itm:nqlg_numbering]{(QL1)} breaks simultaneity during the stage $1$ transitions, while a violation of \hyperref[itm:nqlg_length]{(QL3)} breaks simultaneity during the stage $2$ transitions.
	In the full proof, we show that agents that can detect these faults exist in any non-NQLG.
\end{claimproof}

For part \hyperref[itm:nqlg_2]{(2)}, we show that for any $\mathtt{dA\texttt{\$}f}$-distributed automaton $A$, running on an NLG $G \in \mathsf{NLG}$ and running on some NQLG $\tilde{G} \in \mathsf{NQLG}$ of the same length are equivalent.
In particular, this holds if we choose $A$ to be equivalent to $A^H$ from \Cref{cor:always_acc/rej}, which is possible by the hierarchy of expressive power.
Consider the runs of $A$ on $G$ and $\tilde{G}$.
As in the intuitive description of NQLGs, $\tilde{G}$ can be seen as arising from $G$ by replicating its nodes.
Since agents of $A$ cannot count ($\mathtt{d}$), they can only distinguish between no predecessor (successor) and at least one predecessor (successor).
It follows that, in the runs of $A$ on $G$ and $\tilde{G}$, an agent at a node $v$ in $G$ and all of its replicas in $\tilde{G}$ visit exactly the same sequence of states.
The next proposition formalizes this.
The proof can be found in \Cref{app:nqlg_eq_nlg}.

\begin{restatable}[Correspondence of NQLGs and NLGs]{proposition}{propNqlgEqNlg}\label{prop:nqlg_eq_nlg}
	Let $\tilde{G} = (\tilde{V},\tilde{E},\tilde{\lambda}) \in \mathsf{NQLG}$ be a numbered quasi-line graph with origin set $O$ and length $n \in \mathbb{N} \setminus \{0\}$, and let $G = (V = \{v_0,\dots,v_{n-1}\},E,\lambda) \in \mathsf{NLG}$ be the numbered line graph of the same length $n$.
	Further, let $M$ be a \emph{$\mathtt{dA\texttt{\$}f}$}-distributed machine, and let $\tilde{\rho} = {(\tilde{C}_i)}_{i \in \mathbb{N}}$ and $\rho = {(C_i)}_{i \in \mathbb{N}}$ be the (unique) fair runs of $M$ on $\tilde{G}$ and $G$, respectively.
	Then, for every $i \in \mathbb{N}$ and every $\tilde{v} \in \tilde{V}$, we have $\tilde{C}_i(\tilde{v}) = C_i(v_{\dist(O,\tilde{v})})$.
\end{restatable}

%% file: figs/nqlg_ex.tex
\begin{tikzpicture}[semithick,
	on grid, every node/.style={minimum width=4.2ex}]
	\def\dv{3ex}
	\def\dh{8ex}
	
	\tikzstyle{node}=[draw,circle]
	
	\node[node] (00) 						{$0$};
	\node[node] (01) [below=2*\dv of 00]	{$0$};
	
	\node[node] (10) [right=\dh of 00]		{$1$};
	\node[node] (11) [below=2*\dv of 10]	{$1$};
	\node[node] (12) [below=2*\dv of 11]	{$1$};
	
	\node[node] (20) [right=\dh of 10]		{$2$};
	\node[node] (21) [below=2*\dv of 20]	{$2$};
	
	\node[node] (30) [right=\dh of 20]		{$0$};
	
	\node[node] (40) [right=\dh of 30]		{$1$};
	\node[node] (41) [below=2*\dv of 40]	{$1$};
	
	\draw
		(00) -- (10)
		(00) -- (11)
		(01) -- (10)
		(01) -- (11)
		(01) -- (12)
		
		(10) -- (20)
		(11) -- (20)
		(12) -- (21)
		
		(20) -- (30)
		(21) -- (30)
		
		(30) -- (40)
		(30) -- (41)
		;
\end{tikzpicture}

%% file: chapters/snowballfight.tex
We turn our attention to the class $\mathtt{Da\texttt{\$}f}$.
Like $\mathtt{dA\texttt{\$}f}$-automata, these cannot decide $\mathsf{NLG}$ either, but the reason is more fundamental:
Distributed automata with halting acceptance ($\mathtt{a}$) cannot distinguish between a circular graph with correct numbering modulo $3$ (we will call this a \emph{numbered circular graph} (NCG)), and a sufficiently long NLG.
This can be proven analogously to \cite[Proposition 12]{classification}, as we will now sketch.

Let $A = (M,\Sigma)$ be a distributed automaton with halting acceptance ($\mathtt{a}$) that rejects an NCG $Z$.
By the consistency condition, any fair run of $M$ on $Z$ is rejecting.
Let $\rho = (C_i)_{i \in \mathbb{N}}$ be such a fair run, induced by a fair schedule $\sigma$, and $I \in \mathbb{N}$ such that $C_I$ is rejecting.
We construct an NLG $Z'$ by deleting an edge from $Z$ between a pair of nodes with numbering $2$ and $0$, and then concatenate $2I+1$ copies of $Z'$ to obtain the NLG $L$.
Consider the run of $M$ on $L$ induced by a schedule $\sigma'$ that replicates the first $I$ selections of $\sigma$ on all copies of $Z'$ (note that this can easily be turned into a fair schedule as we only fix finitely many selections).
It is easy to see that, for an agent at distance $k$ from the two endpoints of $L$, at least the first $k$ transitions are exactly those that the agent at the corresponding original node in $Z$ undergoes.
The agents of $L$ in the middle copy of $Z'$ will therefore reach rejecting states after at most $I$ transitions and halt due to halting acceptance ($\mathtt{a}$).
By the consistency condition, $A$ thus has to reject the NLG $L$ too.

This shows that an automaton with halting acceptance ($\mathtt{a}$) that accepts all NLGs necessarily also accepts NCGs.
However, NCGs lack an origin node, which is essential for our TM simulation as that is where the TM head starts.
Without an origin node, there is no simulated TM head, breaking the simulation%
\footnote{Requiring that initially exactly one node has a special label, say $h$, modelling the head does not work either.
	Again, analogously to \cite[Proposition 12]{classification}, one can show that an automaton with halting acceptance ($\mathtt{a}$) cannot accept all graphs with exactly one $h$-labelled node and reject all graphs with none.}.

To solve this problem, we observe that for our purpose of representing arbitrarily long TM tape sections, we do not need to accept all NLGs; rather, it suffices to accept infinitely many NLGs.
This guarantees that for every TM that halts (after finitely many steps), we accept an NLG that is long enough to simulate this run correctly.
We will thus construct a $\mathtt{Da\texttt{\$}f}$-automaton that accepts infinitely many, but not all, NLGs.
For this purpose, we introduce a new family of graphs, which can be seen as a subfamily of NLGs.
In the case of NLGs and NQLGs in \Cref{sec:nlg,sec:nqlg}, we first defined the families, and then constructed automata recognizing them.
In this section, we proceed differently: we construct the automaton first, thereby implicitly defining a family of graphs as the automaton's accepted language.

Intuitively, the automaton lets the agents engage in a \emph{snowball fight}!
For this, we augment the labelling of all agents with a direction, positive or negative, modelling the direction the agent is facing at the beginning of the fight.
Additionally, every other agent initially holds a snowball.
The labelling alphabet hence is $\Lambda^L = \mathbb{Z}_3 \times \{-1,+1\} \times \{0,1\}$: the first component is the numbering, the second is the direction, and the third indicates whether the agent starts with a snowball.
A \emph{snowball fight NLG} (SFNLG) is a $\Lambda^L$-labelled that becomes an NLG after projecting all labels onto their first component; SFNCGs are defined analogously.
We denote the family of SFNLGs by $\mathsf{SFNLG}$.
For example, the third graph of \Cref{fig:snowball_harmonious} is an SFNLG of length 7.
The agents carry out a snowball fight following a fixed set of rules.
Agents can throw snowballs, and catch or get hit by snowballs thrown at them.
Throughout the course of the fight, the number of snowballs will decrease.
The graph is rejected if an agent gets hit, and is accepted if this does not happen until eventually no snowballs are left.
We will show that the former happens for all SFNCGs, while the latter happens in infinitely many SFNLGs.
More details follow below.
The formal description can be found in \Cref{app:snowball}.

\begin{construction}[Snowball Fight!\ (sketch)]\label{pre_cstr:snowball}
	We construct a $\mathtt{Da\texttt{\$}f}$-distributed automaton with the labelling alphabet $\Lambda^L$ as described above.
	The set of states comprises $\Lambda^L$ and some auxiliary states: the accepting state $\checked$, the rejecting error state $\bot$, and a state $\square$ to indicate the intention to accept.
	
	We give an informal description of the behaviour of the automaton.
	First, the automaton uses counting detection ($\mathtt{D}$) as in \Cref{lem:nlg} to check that all nodes have at most degree $2$ and the numbering is correct, producing a $\bot$ state if the check fails.
	This already guarantees that the automaton can only accept SFNLGs or SFNCGs.
	Further, each agent ensures that either itself or both of its neighbours are starting with a snowball.
	Then, the automaton performs the snowball fight, following four simple rules (and some special rules for the agents at the endpoints of an SFNLG, which we will state separately):
	\begin{bracketenumerate}
		\setcounter{enumi}{-1}
		\item If an agent $v$ is holding a snowball, $v$ throws the snowball in the direction it is facing.\label{itm:sf_rule_throw}
		\item If exactly one snowball is thrown at an agent $v$ and $v$ faces towards it, then $v$ catches the snowball and turns around.\label{itm:sf_rule_catch}
		\item If exactly one snowball is thrown at an agent $v$ and $v$ faces away from it, then $v$ gets hit and transitions to the $\bot$ state.\label{itm:sf_rule_hit}
		\item If two snowballs are thrown at an agent $v$, then $v$ catches both of them, merges them into one, and turns around.\label{itm:sf_rule_double}
	\end{bracketenumerate}
	The possible scenarios that can arise from these rules are illustrated in \Cref{fig:snowball_tra}.
	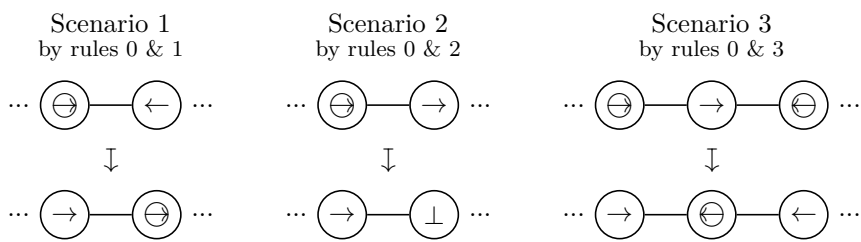
\begin{figure}[h]
		\centering
		\input{figs/snowball_tra.tex}%
		\caption{The three scenarios arising from rules \ref{itm:sf_rule_throw} to \ref{itm:sf_rule_double} --- The arrow in each node indicates the direction the agent is facing.
			A circle indicates that the agent is holding a snowball.
			Note that scenario 2 assumes that we are not in scenario 3, otherwise rule \ref{itm:sf_rule_hit} would not apply.}
		\label{fig:snowball_tra}
	\end{figure}
	Lastly, the agents at the endpoints of SFNLGs can throw a snowball into the void (as they only have one neighbour).
	If this happens at the origin node, it indicates its intention to accept by a $\square$ state; otherwise, this produces a $\bot$ state.
	Intuitively, the $\square$ state means that the agent wants to accept (and thus halt) but cannot do so yet, as it first has to ensure that no other agent has halted in the rejecting $\bot$ state.
	For this, the $\square$ state propagates until it either reaches the other end of the SFNLG and turns into an accepting $\checked$ state, or it is intercepted by a rejecting $\bot$ state.
	Both the $\checked$ and the $\bot$ state propagate unconditionally, ultimately causing acceptance or rejection, respectively.
\end{construction}

We show that \Cref{pre_cstr:snowball} only accepts SFNLGs, and accepts infinitely many SFNLGs, which is sufficient for our purpose.
We split this up into three lemmas:
\begin{bracketenumerate}
	\item If at least one rejecting $\bot$ state occurs, the graph is rejected.\label{itm:sf_error}
	\item If no $\bot$ state occurs, the graph is in $\mathsf{SFNLG}$ and it is accepted.\label{itm:sf_nerror}
	\item There are infinitely many SFNLGs that get accepted.\label{itm:sf_existence}
\end{bracketenumerate}
These lemmas are proved as  \Cref{lem:sf_error,lem:sf_nerror,lem:sf_existence} in \Cref{app:snowball}.
We sketch the proof ideas for the first two and present the relevant graph construction for the third.

\begin{claimproof}[\text{Proof sketch of \hyperref[itm:sf_error]{(1)}}]
	We have to show that after an error state occurred, no accepting $\checked$ state can arise any more.
	Notice that for a $\checked$ state to be produced, the $\square$ state has to propagate from the origin node to the other end of the graph.
	Further notice that after an agent takes on the $\square$ state, it cannot be the first to produce an $\bot$ state, it can only adopt a $\bot$ state from neighbouring agents.
	So, the first time that an error occurs has happen at an agent that has not yet indicated its intention to accept by a $\square$ state.
	In that case however, the $\bot$ state would intercept the propagation of the $\square$ state, preventing it from reaching the other end of the graph to produce an accepting $\checked$ state.
\end{claimproof}

\begin{claimproof}[\text{Proof sketch of \hyperref[itm:sf_nerror]{(2)}}]
	For no error to occur the numbering has to be correct, so the given graph can only be an SFNLG or SFNCG.
	Intuition suggests that on a finite graph, the snowballs should eventually meet and get merged, ultimately decreasing their number to one.
	This is indeed correct.
	In an SFNCG, this last snowball inevitably hits an agent, producing an error by rule \ref{itm:sf_rule_hit} (at the latest after doing a whole round of the circular graph and turning all agents to face the same direction).
	As we assumed no errors to occur, the given graph has to be in $\mathsf{SFNLG}$.
	The last snowball eventually reaches the origin node, producing a $\square$ state, which propagates through the graph unhindered and produces an accepting $\checked$ state.
\end{claimproof}

We conclude that \Cref{pre_cstr:snowball} accepts or rejects every graph and is therefore indeed a distributed automaton.
Moreover, any accepted graph has to be in $\mathsf{SFNLG}$.
Lastly, we construct an infinite (non-exhaustive) family of accepted SFNLGs.

\begin{claimproof}[\text{Construction for \hyperref[lem:sf_existence]{(3)}}]
	We sketch the iterative construction of the direction and snowball labelling (see \Cref{fig:snowball_harmonious}):
	We start with a single agent facing left and holding a snowball.
	Clearly, the last (and only) snowball will get thrown into the void to the left.
	To construct the $(n+1)^\text{th}$ iteration of the construction, we take two copies of the $n^\text{th}$ construction, mirror one of them and connect both to a new right-facing agent such that the last two snowballs, one from each copy of the $n^\text{th}$ construction, meet and merge at this new middle agent.
	From there, the one remaining snowball will be passed to the left until it gets thrown into the void to the left.
	\begin{figure}[h]
		\centering
		\input{figs/snowball_harmonious.tex}%
		\caption{Iterative construction of the labelling of an accepted SFNLG of length $7$ and its accepting run --- The first three rows show the construction of the SFNLG; the other rows show how the snowball fight is performed without errors.}
		\label{fig:snowball_harmonious}
	\end{figure}
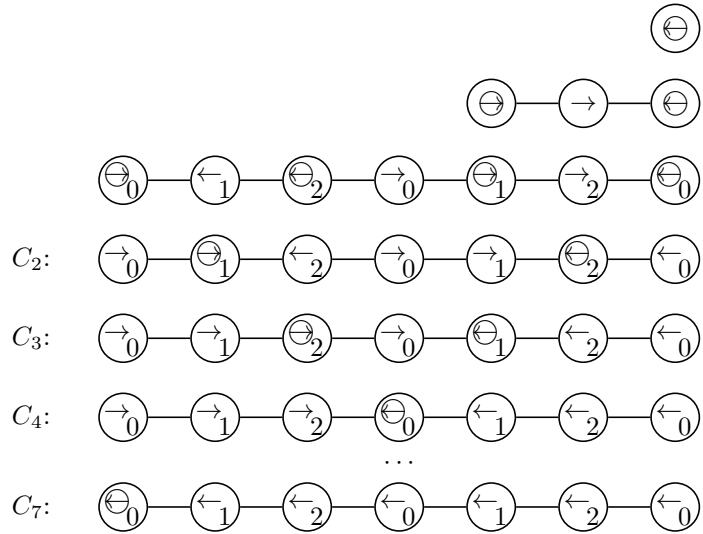
	
	By numbering the nodes such that the leftmost node is the origin node, this ultimately causes acceptance for any iteration of the construction.
	The $n^\text{th}$ construction has length $2^n-1$, yielding infinitely many accepted SFNLGs.
\end{claimproof}

%% file: figs/snowball_tra.tex
\def\dh{8ex}

\tikzstyle{node}=[draw, circle, minimum width=4.2ex, inner sep=0pt]

\begin{tikzpicture}[semithick, on grid]
	\node			(i_0)							{$...$};
	\node[node]		(i0)	[right=0.5*\dh of i_0]	{$\rsnowball$};
	\node[node]		(i1)	[right=\dh of i0]		{$\lface$};
	\node			(i_1)	[right=0.5*\dh of i1]	{$...$};
	
	\draw
		(i0) -- (i1)
		;
	
	\node 			(to)	[below=0.6*\dh of i_0, xshift=\dh]	{\mapsdown};
	
	\node			(o_0)	[below=1.2*\dh of i_0]	{$...$};
	\node[node]		(o0)	[right=0.5*\dh of o_0]	{$\rface$};
	\node[node]		(o1)	[right=\dh of o0]		{$\rsnowball$};
	\node			(o_1)	[right=0.5*\dh of o1]	{$...$};
	
	\draw
		(o0) -- (o1)
		;
	
	\node		(title)		[above=1.5*\dh of to]		{Scenario 1};
	\node		(subtitle)	[below=0.3*\dh of title]	{\footnotesize by rules \ref{itm:sf_rule_throw} \& \ref{itm:sf_rule_catch}};
\end{tikzpicture}\qquad
\begin{tikzpicture}[semithick, on grid]
\node			(i_0)							{$...$};
\node[node]		(i0)	[right=0.5*\dh of i_0]	{$\rsnowball$};
\node[node]		(i1)	[right=\dh of i0]		{$\rface$};
\node			(i_1)	[right=0.5*\dh of i1]	{$...$};

\draw
(i0) -- (i1)
;

\node (to)				[below=0.6*\dh of i_0, xshift=\dh]	{\mapsdown};

\node			(o_0)	[below=1.2*\dh of i_0]	{$...$};
\node[node]		(o0)	[right=0.5*\dh of o_0]	{$\rface$};
\node[node]		(o1)	[right=\dh of o0]		{$\bot$};
\node			(o_1)	[right=0.5*\dh of o1]	{$...$};

\draw
(o0) -- (o1)
;

\node		(title)		[above=1.5*\dh of to]		{Scenario 2};
\node		(subtitle)	[below=0.3*\dh of title]	{\footnotesize by rules \ref{itm:sf_rule_throw} \& \ref{itm:sf_rule_hit}};
\end{tikzpicture}\qquad
\begin{tikzpicture}[semithick, on grid]
	\node			(i_0)							{$...$};
	\node[node]		(i0)	[right=0.5*\dh of i_0]	{$\rsnowball$};
	\node[node]		(i1)	[right=\dh of i0]		{$\rface$};
	\node[node]		(i2)	[right=\dh of i1]		{$\lsnowball$};
	\node			(i_2)	[right=0.5*\dh of i2]	{$...$};
	
	\draw
		(i0) -- (i1)
		(i1) -- (i2)
		;
	
	\node (to)				[below=0.6*\dh of i_0, xshift=1.5*\dh]	{\mapsdown};
	
	\node			(o_0)	[below=1.2*\dh of i_0]	{$...$};
	\node[node]		(o0)	[right=0.5*\dh of o_0]	{$\rface$};
	\node[node]		(o1)	[right=\dh of o0]		{$\lsnowball$};
	\node[node]		(o2)	[right=\dh of o1]		{$\lface$};
	\node			(o_1)	[right=0.5*\dh of o2]	{$...$};
	
	\draw
		(o0) -- (o1)
		(o1) -- (o2)
		;
	
	\node		(title)		[above=1.5*\dh of to]		{Scenario 3};
	\node		(subtitle)	[below=0.3*\dh of title]	{\footnotesize by rules \ref{itm:sf_rule_throw} \& \ref{itm:sf_rule_double}};
\end{tikzpicture}

%% file: figs/snowball_harmonious.tex
\def\dh{8ex}
\def\das{0.75ex}
\def\dn{1.2ex}

\tikzstyle{node}=[draw, circle, minimum width=4.2ex, inner sep=0pt]

\pgfmathsetmacro{\num}{7}

\begin{tikzpicture}[semithick, on grid]
	\node (C) {\phantom{$C_1$:}};
	
	\def\styles{{"", "", "", "", "", "", "node"}}
	\def\arrowsnowball{{"\phantom{\rsnowball}", "\phantom{\lface}", "\phantom{\lsnowball}", "\phantom{\rface}", "\phantom{\rsnowball}", "\phantom{\rface}", "\lsnowball"}}
	
	\pgfmathparse{\styles[0]}
	\node[\pgfmathresult] (0) [right=\dh of C] {\pgfmathparse{\arrowsnowball[0]}$\pgfmathresult$};
	\foreach \i in {1,...,\the\numexpr\num-1\relax} {
		\pgfmathtruncatemacro{\label}{mod(\i, 3)};
		\pgfmathparse{\styles[\i]}
		\node[\pgfmathresult] (\i) [right=\dh of \the\numexpr\i-1\relax] {\pgfmathparse{\arrowsnowball[\i]}$\pgfmathresult$};
	}
\end{tikzpicture}%
\vspace{2ex}

\begin{tikzpicture}[semithick, on grid]
	\node (C) {\phantom{$C_1$:}};
	
	\def\styles{{"", "", "", "", "node", "node", "node"}}
	\def\arrowsnowball{{"\phantom{\rsnowball}", "\phantom{\lface}", "\phantom{\lsnowball}", "\phantom{\rface}", "\rsnowball", "\rface", "\lsnowball"}}
	
	\pgfmathparse{\styles[0]}
	\node[\pgfmathresult] (0) [right=\dh of C] {\pgfmathparse{\arrowsnowball[0]}$\pgfmathresult$};
	\foreach \i in {1,...,\the\numexpr\num-1\relax} {
		\pgfmathtruncatemacro{\label}{mod(\i, 3)};
		\pgfmathparse{\styles[\i]}
		\node[\pgfmathresult] (\i) [right=\dh of \the\numexpr\i-1\relax] {\pgfmathparse{\arrowsnowball[\i]}$\pgfmathresult$};
	}
	
	\foreach \i in {5,6} {
		\draw (\the\numexpr\i-1\relax) -- (\i);
	}
\end{tikzpicture}%
\vspace{2ex}

\begin{tikzpicture}[semithick, on grid]
	\node (C) {\phantom{$C_1$:}};
	
	\def\arrowsnowball{{"\rsnowball", "\lface", "\lsnowball", "\rface", "\rsnowball", "\rface", "\lsnowball"}}
	
	\node[node] (0) [right=\dh of C] {};
	\node		(as0) [above left=\das of 0] {\pgfmathparse{\arrowsnowball[0]}$\pgfmathresult$};
	\node		(n0) [below right=\dn of 0] {$0$};
	\foreach \i in {1,...,\the\numexpr\num-1\relax} {
		\pgfmathtruncatemacro{\label}{mod(\i, 3)};
		\node[node] (\i) [right=\dh of \the\numexpr\i-1\relax] {};
		\node		(as\i) [above left=\das of \i] {\pgfmathparse{\arrowsnowball[\i]}$\pgfmathresult$};
		\node		(n\i) [below right=\dn of \i] {$\label$};
	}
	
	\foreach \i in {1,...,\the\numexpr\num-1\relax} {
		\draw (\the\numexpr\i-1\relax) -- (\i);
	}
\end{tikzpicture}%
\vspace{2ex}

\begin{tikzpicture}[semithick, on grid]
	\node (C) {$C_2$:};
	
	\def\arrowsnowball{{"\rface", "\rsnowball", "\lface", "\rface", "\rface", "\lsnowball", "\lface"}}
	
	\node[node] (0) [right=\dh of C] {};
	\node		(as0) [above left=\das of 0] {\pgfmathparse{\arrowsnowball[0]}$\pgfmathresult$};
	\node		(n0) [below right=\dn of 0] {$0$};
	\foreach \i in {1,...,\the\numexpr\num-1\relax} {
		\pgfmathtruncatemacro{\label}{mod(\i, 3)};
		\node[node] (\i) [right=\dh of \the\numexpr\i-1\relax] {};
		\node		(as\i) [above left=\das of \i] {\pgfmathparse{\arrowsnowball[\i]}$\pgfmathresult$};
		\node		(n\i) [below right=\dn of \i] {$\label$};
	}
	
	\foreach \i in {1,...,\the\numexpr\num-1\relax} {
		\draw (\the\numexpr\i-1\relax) -- (\i);
	}
\end{tikzpicture}%
\vspace{2ex}

\begin{tikzpicture}[semithick, on grid]
	
	\node (C) {$C_3$:};
	
	\def\arrowsnowball{{"\rface", "\rface", "\rsnowball", "\rface", "\lsnowball", "\lface", "\lface"}}
	
	\node[node] (0) [right=\dh of C] {};
	\node		(as0) [above left=\das of 0] {\pgfmathparse{\arrowsnowball[0]}$\pgfmathresult$};
	\node		(n0) [below right=\dn of 0] {$0$};
	\foreach \i in {1,...,\the\numexpr\num-1\relax} {
		\pgfmathtruncatemacro{\label}{mod(\i, 3)};
		\node[node] (\i) [right=\dh of \the\numexpr\i-1\relax] {};
		\node		(as\i) [above left=\das of \i] {\pgfmathparse{\arrowsnowball[\i]}$\pgfmathresult$};
		\node		(n\i) [below right=\dn of \i] {$\label$};
	}
	
	\foreach \i in {1,...,\the\numexpr\num-1\relax} {
		\draw (\the\numexpr\i-1\relax) -- (\i);
	}
\end{tikzpicture}%
\vspace{2ex}

\begin{tikzpicture}[semithick, on grid]
	
	\node (C) {$C_4$:};
	
	\def\arrowsnowball{{"\rface", "\rface", "\rface", "\lsnowball", "\lface", "\lface", "\lface"}}
	
	\node[node] (0) [right=\dh of C] {};
	\node		(as0) [above left=\das of 0] {\pgfmathparse{\arrowsnowball[0]}$\pgfmathresult$};
	\node		(n0) [below right=\dn of 0] {$0$};
	\foreach \i in {1,...,\the\numexpr\num-1\relax} {
		\pgfmathtruncatemacro{\label}{mod(\i, 3)};
		\node[node] (\i) [right=\dh of \the\numexpr\i-1\relax] {};
		\node		(as\i) [above left=\das of \i] {\pgfmathparse{\arrowsnowball[\i]}$\pgfmathresult$};
		\node		(n\i) [below right=\dn of \i] {$\label$};
	}
	
	\foreach \i in {1,...,\the\numexpr\num-1\relax} {
		\draw (\the\numexpr\i-1\relax) -- (\i);
	}
	
	\node[below=0.5*\dh of 3] {$\,\cdots$};
\end{tikzpicture}%

\begin{tikzpicture}[semithick, on grid]
	
	\node (C) {$C_7$:};
	
	\def\arrowsnowball{{"\lsnowball", "\lface", "\lface", "\lface", "\lface", "\lface", "\lface"}}
	
	\node[node] (0) [right=\dh of C] {};
	\node		(as0) [above left=\das of 0] {\pgfmathparse{\arrowsnowball[0]}$\pgfmathresult$};
	\node		(n0) [below right=\dn of 0] {$0$};
	\foreach \i in {1,...,\the\numexpr\num-1\relax} {
		\pgfmathtruncatemacro{\label}{mod(\i, 3)};
		\node[node] (\i) [right=\dh of \the\numexpr\i-1\relax] {};
		\node		(as\i) [above left=\das of \i] {\pgfmathparse{\arrowsnowball[\i]}$\pgfmathresult$};
		\node		(n\i) [below right=\dn of \i] {$\label$};
	}
	
	\foreach \i in {1,...,\the\numexpr\num-1\relax} {
		\draw (\the\numexpr\i-1\relax) -- (\i);
	}
\end{tikzpicture}

%% file: chapters/conclusion.tex
We have initiated the study of verification problems for the classes of distributed automata introduced in \cite{classification}.
We have shown that the emptiness problem, a fundamental verification problem, is undecidable or trivially decidable for all classes of \cite{classification}.
Since distributed automata are closed under intersection (\Cref{lem:intersection}), this means that the safety problem -- given an automaton and a class $\mathcal{D}$ of ``dangerous'' graphs, does the automaton accept some graph of $\mathcal{D}$? -- is undecidable whenever $\mathcal{D}$ is recognized by some automaton.

Our undecidability proofs simulate the execution of a Turing machine on blank tape by means of a distributed automaton running on families of labelled graphs ($\mathsf{NLG}$, $\mathsf{NQLG}$, and a subfamily of $\mathsf{SFNLG}$) having both a specific structure and a specific labelling.
The proofs break down if we restrict ourselves to automata that only allow unlabelled graphs as inputs, so they can only decide purely structural properties of the input graph, or to automata that only decide labelling properties.
The decision power of the latter was studied by Czerner \emph{et al.} in \cite{decision_power}.
In future work, we plan to study these two special cases.
We conjecture that the emptiness problem for automata deciding structural properties remains undecidable, but becomes decidable for some of the classes of \cite{decision_power}.

%% file: chapters/appendix_groundwork.tex
\subsection{Proof of \Cref{lem:intersection}}\label{app:intersection}
\input{chapters/appendix_intersection.tex}

\subsection{Proofs of \Cref{sec:nlg,sec:tm_head}}
\input{chapters/appendix_warmp.tex}

\subsection{Proofs of \Cref{sec:nqlg}}
\input{chapters/appendix_nqlg.tex}

%% file: chapters/appendix_intersection.tex
\intersection*

We give the construction for the second claim and prove its correctness.
This immediately implies that distributed automata are closed under intersection.
We first consider distributed automata with acceptance by stable consensus ($\mathtt{A}$), where the construction is slightly simpler.
Then, we discuss the necessary changes for distributed automata with halting acceptance ($\mathtt{a}$).

\begin{construction}[Product of Distributed Machines ($\mathtt{A}$)]\label{cstr:product_A}
	We are given a distributed automaton $A_1 = (M_1 = (Q_1,\delta_{10},\delta_1,Y_1,N_1), \Sigma)$ with counting bound $\beta_1$ and an $L(A_1)$-distributed automaton $A_2 = (M_2 = (Q_2,\delta_{20},\delta_2,Y_2,N_2), \Sigma)$ with counting bound $\beta_2$.
	Both $A_1,A_2$ are in the class $x\mathtt{A}wz$ with $xwz \in \{\mathtt{d},\mathtt{D}\}\{\texttt{\$},\mathtt{s},\mathtt{S}\}\{\mathtt{f},\mathtt{F}\}$ and operate on the same labelling alphabet $\Lambda$.
	We define an $x\mathtt{A}wz$-distributed machine $M = (Q,\delta_0,\delta,Y,N)$ with:
	\begin{itemize}
		\item The labelling alphabet is $\Lambda$.
		\item The counting bound is $\max\{\beta_1,\beta_2\}$.
		\item $Q = Q_1 \times Q_2$.
		As in the familiar DFA product construction, the states are pairs of states from the two machines $M_1,M_2$.
		\item $\delta_0 \colon {\Lambda \to Q, l \mapsto \big(\delta_{10}(l),\delta_{20}(l)\big)}$.
		The pair of states is initialized with the initialization functions of $M_1,M_2$ in the respective components.
		\item $\delta \colon {Q \times {[\beta]}^Q \to Q, \big((q_1,q_2),\mathcal{N})\big) \mapsto \big(\delta_1(q_1,\mathcal{N}(\cdot,*)),\delta_2(q_2,\mathcal{N}(*.\cdot))\big)}$.
		Again, as for DFAs, a transition in $M$ executes the transitions of $M_1,M_2$ in the respective components.
		Note that we use the notation $\mathcal{N}(\cdot,*)(q_1) = \mathcal{N}(q_1,*)$ (analogously $\mathcal{N}(*,\cdot)$) with the $*$-notation as introduced in the preliminaries (\Cref{sec:prelim}).
		\item The accepting states are $Y = Y_1 \times Y_2$, i.e., the states where both $M_1$ and $M_2$ accept, while the rejecting states are $N = (N_1 \times Q_2) \cup (Y_1 \times N_2)$, i.e., the states where $M_1$ rejects, or $M_1$ accepts and $M_2$ rejects.
	\end{itemize}
	Intuitively, $M$ executes $M_1$ in the first component of the agents' states and $M_2$ in the second.
	A graph is accepted if both machines accept, and rejected if either $M_1$ rejects, or $M_1$ accepts and $M_2$ rejects.
\end{construction}
	
Note that the construction would also work if we chose $(N_1 \times Q_2) \cup (Q_1 \times N_2)$ as rejecting states.
However, our choice of $N$ is made to emphasize that a decision of $A_2$ without $A_1$ accepting is somewhat meaningless, as we cannot guarantee consistency of fair runs of $M_2$ on graphs outside $L(A_1)$.
$A_1$ on the other hand, is a distributed automaton and satisfies the consistency condition on all graphs.
Only once $M_1$ accepts, we can draw reliable conclusions from $M_2$'s decision.
If $M_1$ rejects, $M$ rejects anyway, independent of what $M_2$ does.
Furthermore, choosing $N$ in this way, will be important for \Cref{cstr:product_a} -- the corresponding construction for automata with halting acceptance ($\mathtt{a}$).

\begin{proof}[Proof of \Cref{lem:intersection} for $\mathtt{A}$]\label{prf:product_A}
	Let $A \coloneq (M,\Sigma)$, where $M$ is the $x\mathtt{A}wz$-distributed machine from \Cref{cstr:product_A} and $\Sigma$ a matching scheduler.
	Further, let $G = (V,E,\lambda)$ be any $\Lambda$-labelled graph and $\sigma$ a fair schedule.
	Since $A,A_1,A_2$ all use a scheduler of the same class, $\sigma$ induces fair runs $\rho = (C_i)_{i \in \mathbb{N}}, \rho_1 = (C_{1i})_{i \in \mathbb{N}}, \rho_2 = (C_{2i})_{i \in \mathbb{N}}$ of $M,M_1,M_2$, respectively, on $G$.
	By construction of $M$, we clearly have $C_i(v) = (C_{1i}(v),C_{2i}(v))$ for all $i \in \mathbb{N}, v \in V$.

	Since the distributed automaton $A_1$ satisfies the consistency condition on all graphs, it unambiguously accepts or rejects $G$.
	If $A_1$ accepts (rejects) $G$, there is an index $I_1 \in \mathbb{N}$, such that $C_{1i}$ is accepting (rejecting) for all $i \geq I_1$.
	We thus have
	$$C_i(V) \subseteq C_{1i}(V) \times C_{2i}(V) \subseteq \begin{cases}
		Y_1 \times Q_2 & \text{if $A_1$ accepts $G$,}\\
		N_1 \times Q_2 \subseteq N & \text{if $A_1$ rejects $G$}\\ 
	\end{cases}$$
	for all $i \geq I_1$.
	So if $A_1$ rejects $G$, then so does $A$.
	
	Now, assume that $A_1$ accepts $G$, i.e., $G \in L(A_1)$.
	Since $A_2$ is an $L(A_1)$-distributed automaton, it also unambiguously accepts or rejects $G$.
	There hence is an index $I_2$ for $\rho_2$ analogous to $I_1$ for $\rho_1$.
	We thus have
		$$C_i(V) \subseteq C_{1i}(V) \times C_{2i}(V) \subseteq \begin{cases}
		Y_1 \times Y_2 \subseteq Y & \text{if $A_2$ accepts $G$,}\\
		Y_1 \times N_2 \subseteq N & \text{if $A_2$ rejects $G$}\\ 
	\end{cases}$$
	for all $i \geq \max\{I_1,I_2\}$.
	So, if $A_2$ accepts (rejects) $G$, then so does $A$.
	
	With this case exhaustion, we conclude that $A$ accepts $G$ if{}f $A_1,A_2$ accept $G$, and rejects otherwise, i.e., $L(A) = L(A_1) \cap L(A_2)$.
	In particular, $A$ satisfies the consistency condition on all graphs, making it an $x\mathtt{A}wz$-distributed automaton.
\end{proof}

\begin{construction}[Product of Distributed Machines ($\mathtt{a}$)]\label{cstr:product_a}
	We are given $A_1,A_2$ as before, except this time they are in the class $x\mathtt{a}wz$.
	Based on $M = (Q,\delta_0,\delta,Y,N)$ from \Cref{cstr:product_A}, we make an adjustment.
	We define an $x\mathtt{a}wz$-distributed machine $M' = (Q,\delta_0,\delta',Y,N)$ with:
	\begin{itemize}
		\item The labelling alphabet $\Lambda$, the counting bound $\beta$, the set of states $Q$, the initialization function $\delta_0$, and the accepting and rejecting sets remain unchanged.
		\item We define $\delta' \colon Q \times {[\beta]}^{Q} \to Q$:
		$$(q,\mathcal{N}) \mapsto \begin{cases}
			q & \text{if } q \in N_1 \times Q_2,\\
			\delta(q,\mathcal{N}) & \text{else.}
		\end{cases}$$
	\end{itemize}
	Intuitively, we run $M$, unless an agent $v$ is in a rejecting state in $N_1 \times Q_2 \subseteq N$; in that case, $\delta'$ forces $v$ to halt in order to comply with halting acceptance ($\mathtt{a}$).
	All other accepting or rejecting states are in $Y_1 \times Y_2 = Y$ or $Y_1 \times N_2 \subseteq N$ and thus halting anyway, as $M_1,M_2$ employ halting acceptance ($\mathtt{a}$).
\end{construction}

\begin{proof}[Proof of \Cref{lem:intersection} for $\mathtt{a}$]
	Let $A' \coloneq (M',\Sigma)$, where $M'$ is the $x\mathtt{a}wz$-distributed machine from \Cref{cstr:product_a} and $\Sigma$ a matching scheduler.
	Applying \Cref{cstr:product_A} to $A_1,A_2$, yields an $x\mathtt{A}wz$-distributed automaton $A$ with $L(A) = L(A_1) \cap L(A_2)$ by the previous proof.
	In contrast to $A'$, $A$ does not belong to the class $x\mathtt{a}wz$, as the second component of rejecting states in $N_1 \times Q_2 \subseteq N$ can still change.
	It hence suffices to show that $A'$ accepts if{}f $A$ accepts, and rejects otherwise.
	For this, let $G = (V,E,\lambda)$ be any $\Lambda$-labelled graph and $\sigma$ a fair schedule.
	Since $A,A',A_1$ all use a scheduler of the same class, $\sigma$ induces fair runs $\rho = (C_i)_{i \in \mathbb{N}}, \rho' = (C_i')_{i \in \mathbb{N}}, \rho_1 = (C_{1i})_{i \in \mathbb{N}}$ of $M,M',M_1$, respectively, on $G$.
	
	If no state in $N_1 \times Q_2 \subseteq N$ ever occurs in $\rho$, then $\delta'$ is defined such that $\rho = \rho'$.
	So $\rho'$ is accepting if{}f $\rho$ is accepting, and rejecting otherwise.
	
	Now, assume that an agent $v_0 \in V$ of $A$ eventually reaches a state $(n_1,q_2) \in N_1 \times Q_2 \subseteq N$ in $\rho$, i.e., there exists an $i_0 \in \mathbb{N}$ with $C_{i_0}(v_0) = (n_1,q_2)$.
	This implies $C_{1i_0}(v_0) = n_1 \in N_1$, making $v_0$ as an agent of $A_1$ halt in a rejecting state, as $A_1$ is a distributed automaton with halting acceptance ($\mathtt{a}$).
	$\rho_1$ can therefore only be rejecting, implying that $\rho$ is rejecting too.
	We claim that this implies $\rho'$ is rejecting as well.
	Let $\rho'_{(1)} = ({C_i'(\cdot)}_1)$ denote the projection onto the first component of each state in each configuration of $\rho'$.
	We show by induction on $i$ that $\rho'_{(1)} = \rho_1$.
	We have ${C_0'(v)}_1 = {\delta_0(\lambda(v))}_1 = \delta_{10}(\lambda(v)) = C_{10}(v)$ for all $v \in V$, establishing the base case.
	Now, assume $C_i'(v)_1 = C_{1i}(v)$ for all $v \in V$ to hold for an arbitrary but fixed $i \in \mathbb{N}$.
	Let $v_0 \in V$ be any agent.
	Applying the induction hypothesis to $v_0$ and its neighbours, we have ${C_i'(v_0)}_1 = C_{1i}(v_0) \eqcolon q_1$ and $\mathcal{N}_{v_0}^{C_i'}(\cdot,*) = \mathcal{N}_{v_0}^{C_{1i}} \eqcolon \mathcal{N}_1$.
	If $q_1 \in N_1$, $\delta'$ forces $v_0$ as an agent of $A'$ to halt, yielding ${C_{i+1}'(v_0)}_1 = q_1$.
	Similarly, $v_0$ halts as an agent of $A_1$, yielding $C_{1(i+1)}(v_0) = q_1$.
	On the other hand, if $q_1 \notin N_1$, $\delta'$ just applies $\delta$, yielding
	\begin{flalign*}
		&&{C_{i+1}'(v_0)}_1	&= {\delta\big(C_i'(v_0),\mathcal{N}_{v_0}^{C_i'}\big)}_1 = \delta_1\big({C_i'(v_0)}_1,\mathcal{N}_{v_0}^{C_i'}(\cdot,*)\big)&& \\
		&&					&= \delta_1(q_1,\mathcal{N}_1) = \delta_1\big(C_{1i}(v_0),\mathcal{N}_{v_0}^{C_{1i}}\big) = C_{1(i+1)}(v_0).&&
	\end{flalign*}
	As $v_0 \in V$ was arbitrary, this completes the induction step.
	As $\rho_1$ is rejecting, there is an index $I_1 \in \mathbb{N}$, such that $C_{1i}$ is rejecting for all $i \geq I_1$.
	This implies $C_i'(V)_1 = C_{1i}(V) \subseteq N_1$ and hence $C_i'(V) \subseteq N_1 \times Q_2 \subseteq N$ for all $i \geq I_1$, concluding that $\rho'$ is rejecting.
\end{proof}

%% file: chapters/appendix_warmp.tex
\subsubsection{Proof of \Cref{lem:nlg}}\label{app:nlg}
\input{chapters/appendix_nlg.tex}

\subsubsection{Proof of \Cref{cor:always_acc/rej}}\label{app:tm_head}
\input{chapters/appendix_tm_head.tex}

\subsubsection{Proof of \Cref{thm:reduction} for $\mathtt{DAf}$ and $\mathtt{DAF}$}\label{app:proof_1}
\input{chapters/appendix_theorem_1.tex}

%% file: chapters/appendix_nlg.tex
\lemNlg*

First, we formalize the definition of $\mathsf{NLG}$ and the construction described in the proof sketch of \Cref{lem:nlg}.

\begin{definition2}[Numbered Linear Graphs]\label{def:nlg}
	Let $G = (V,E,\lambda)$ with $\lambda \colon {V \to \mathbb{Z}_3}$.
	We call $G$ a \emph{numbered linear graph} of length $n \in \mathbb{N} \setminus \{0\}$ if{}f
	\begin{description}
		\item[(L1)] $V = \{v_0,\dots,v_{n-1}\}, E = \{\{v_i,v_{i+1}\} \mid i \in [n-2]\}$ ($G$ is linear), and\label{itm:nlg_line}
		\item[(L2)] $\forall i \in [n-1] \colon {\lambda(v_i) = i\bmod{3}}$ (agent numbering modulo 3).\label{itm:nlg_numbering}
	\end{description}
	Note that we are again using $0$-indexing.
	The node $v_0$ that uniquely has numbering $0$ and no left neighbour, that is, no neighbour numbered $2$, is called the \emph{origin node}.
	We denote the family of numbered linear graphs by $\mathsf{NLG}$.
\end{definition2}

\begin{construction}[Deciding Numbered Linear Graphs]\label{cstr:nlg}
	We construct a $\mathtt{DA\texttt{\$}f}$-distributed machine $M^L = (Q^L,\delta_0^L,\delta^L,Y^L,N^L)$ with:
	\begin{itemize}
		\item The labelling alphabet is $\Lambda^L = \mathbb{Z}_3$.
		\item $Q^L = (\Lambda^L \times \{0,1\}) \cup \{\bot\}$.
		Intuitively, the first component is the (static) numbering, and the second component is the agent's current guess whether the graph has an origin node.
		$\bot$ is an error state.
		\item $\delta_0^L \colon {\Lambda^L \to Q^L, l \mapsto (l,0)}.$
		The numbering is taken directly from the labelling, and each agent's initial guess is that there is no origin node.
		\item The accepting states are $Y^L = \mathbb{Z}_3 \times \{1\}$, i.e., the states where the agent guesses that there is an origin node, while the rejecting states are simply the rest $N^L = Q^L \setminus Y^L$, i.e., the states where the agent guesses that there is no origin node and the error state.
	\end{itemize}
	For an agent $v \in V$ and a configuration $C \colon {V \to Q}$, we define the transition function $\delta^L$:
	\begin{equation}
		\begin{aligned}
			\big((n,g),\mathcal{N}_v^C\big) & \mapsto \bot & \text{for } & \mathcal{N}_v^C(n,*) \geq 1\ \vee \\
			&&& \exists m \in \mathbb{Z}_3 \colon {\mathcal{N}_v^C(m,*) \geq 2}, \\
			\big(q,\mathcal{N}_v^C\big) & \mapsto \bot & \text{for } & \mathcal{N}_v^C(\bot) \geq 1,
		\end{aligned}\labelledtag{L-error}
	\end{equation}
	that is, an agent will transition to the error state if it detects a neighbour with the same numbering as itself or two neighbours with the same numbering as each other, since this means that the graph's numbering is incorrect.
	Furthermore, if any neighbour of $v$ already is in the error state, $v$'s next state is $\bot$, so error states propagate through the graph.
	For all the following transitions, we assume that, in addition to the specified conditions being fulfilled, the conditions for the \ref{L-error} transitions are not met:
	\begin{align}
		\big((0,0),\mathcal{N}_v^C\big) & \mapsto (0,1) & \text{for } & \mathcal{N}_v^C(2,*) = 0 \labelledtag{L-origin}
		\intertext{If the numbering indicates that $v$ does not have a left neighbour, $v$ guesses to be the origin node $v_0$.}
		\big((n,0),\mathcal{N}_v^C\big) & \mapsto (n,1) & \text{for } & \mathcal{N}_v^C(*,1) \geq 1, \labelledtag{L-line}
	\end{align}
	that is, positive guesses propagate through the graph.
	Note that this is still under the condition that no \ref{L-error} transitions apply.
		
	If none of the above transitions apply, $v$ remains in its previous state.
	We call this a \emph{silent} transition.
\end{construction}

\begin{proof}[Proof of \Cref{lem:nlg}]
	Let $A^L = (M^L, \Sigma)$, where $M^L$ is \Cref{cstr:nlg} and $\Sigma$ a matching scheduler.
	
	We start with the case that $G = (V = \{v_0,\dots,v_{n-1}\},E,\lambda)$ is, in fact, an NLG of length $n \in \mathbb{N} \setminus \{0\}$.
	Then, no node has a neighbour with the same numbering as itself or two neighbours with the same numbering.
	Therefore, the first \ref{L-error} transition will never be executed and consequently neither will the second.
	Every agent starts with the guess $0$.
	In the first transition, $v_0$ performs the \ref{L-origin} transition, changing its guess to $1$.
	For all $i \in [n-1] \setminus \{0\}$, if $v_{i-1}$ already guesses $1$, the agent $v_i$ will change its guess to $1$ in the next transition by applying \ref{L-line}.
	Inductively, it follows that all agents eventually guess $1$, accepting $G$.
	\medskip
	
	Now, we assume that $G = (V,E,\lambda)$ is not an NLG, so it violates \hyperref[itm:nlg_line]{(L1)}, i.e., being linear, or \hyperref[itm:nlg_numbering]{(L2)}, i.e., agent numbering modulo $3$.
	
	$G$ not being linear could mean that there is at least one node $v \in V$ with at least three neighbours $v^1,v^2,v^3$.
	By the pigeon hole principle, there are at least two nodes among $v,v^1,v^2,v^3$ with the same numbering.
	If $v$ is one of these two nodes, the first condition of the first \ref{L-error} transition is met; if the same numbering is between $v^i,v^j,i \neq j$, the second condition of the first \ref{L-error} transition is met.
	So, in both cases $v$ will transition to the error state.
	The $\bot$ state will propagate through the graph, so that $G$ is rejected.
	
	If every node has at most two neighbours, but $G$ is still not linear, this means that every node has exactly two neighbours.
	If there is an agent $v$ with neighbours $v^1,v^2$ that can perform \ref{L-origin}, the three nodes $v,v^1,v^2$ can only be numbered with $0$ and $1$.
	As above, we can apply the pigeon hole principle and show that $G$ is rejected.
	On the other hand, if no agent can perform the \ref{L-origin} transition, this circumstance will always be the case, as it only depends on the static numbering.
	Since all agents start with guess $0$, there will hence never be a first agent to guess $1$, meaning that all agents stay in rejecting states indefinitely; $G$ is rejected.
	
	Now for the case that $G$ is a linear graph, but the numbering is incorrect.
	If none of the agents can perform \ref{L-origin}, $G$ is rejected as above.
	Otherwise, starting from the node $v_0$ that can perform the \ref{L-origin} transition, we number the nodes in an ascending manner.
	Since the numbering is incorrect, there is a first node $v_i, i \geq 2$ with the wrong numbering $(i-2)\bmod{3}$ or $(i-1)\bmod{3}$.
	As the numbering of $v_{i-2}$ and $v_{i-1}$ is correct, the node $v_{i-1}$ and its two neighbours $v_{i-2},v_i$ are only numbered with $(i-2)\bmod{3}$ and $(i-1)\bmod{3}$.
	The rejection of $G$ again follows using the pigeon hole principle as before.
\end{proof}

%% file: chapters/appendix_tm_head.tex
\simulation*

We present the relevant construction that was sketched in \Cref{sec:tm_head}.
Recall that our TM model introduced in \Cref{sec:prelim} has a tape that is bounded to the left and that the TM head starts on the leftmost cell.

\begin{construction}[Turing Machine Head]\label{cstr:tm_head}
	Given a Turing machine $T = (Q,q_0,F,\Gamma,\Sigma,\square,\delta)$, we construct a $\mathtt{da\texttt{\$}f}$-distributed machine $M^H(T) = (Q^H,\delta_0^H,\delta^H,Y^H,N^H)$ with:
	\begin{itemize}
		\item The labelling alphabet is $\Lambda^H = \mathbb{Z}_3$.
		\item $Q^H = Q' \cup (Q' \times Q \times \mathcal{H}) \cup (\{\circ\} \times Q') \cup \{\checked, \bot\}$ with $Q' = \Gamma \times \Lambda^H$ and $\mathcal{H} = \{H,H_{+1},H_{-1}\}$.
		A state in $Q'$ consists of a symbol from the TM's tape alphabet and the node's numbering.
		If an agent represents the TM head, it additionally saves the current TM state from $Q$, and one out of three TM head states from $\mathcal{H}$, which are to be interpreted as the TM head after a move ($H$), or before a move with the intention to move in the positive/right or negative/left direction ($H_{+1},H_{-1}$), respectively.
		States in $Q'$ can carry an uninitialized-marker $\circ$, which will be explained later.
		$\checked$ and $\bot$ are the accepting and rejecting states, respectively.
		\item $\delta_0^H \colon {\Lambda^H \to Q^H, l \mapsto (\circ,\square,l)}.$
		Every agent is initialized with the uninitialized-marker set and a blank symbol, representing blank tape.
		The numbering is again taken directly from the labelling.
		\item As mentioned, the accepting state is $Y^H = \{\checked\}$, while the rejecting state is $N^H = \{\bot\}$.
	\end{itemize}
	For an agent $v \in V$ and a configuration $C \colon {V \to Q}$, we define the transition function $\delta^H$:
	\begin{align}
		\big((\circ,\square,0),\mathcal{N}_v^C\big) & \mapsto (\square,0,q_0,H) && \text{for } \mathcal{N}_v^C(\circ,*,2) = 0 \labelledtag{H-origin} \\
		\big((\circ,\square,n),\mathcal{N}_v^C\big) & \mapsto (\square,n) && \text{otherwise} \labelledtag{H-init}
	\end{align}
	These two transitions initialize the uninitialized agents $v$ for the simulation.
	If $v$ is the unique origin node, it adopts the TM's initial state $q_0$ and the TM head state $H$, while removing its marker.
	Otherwise, $v$ only removes its uninitialized-marker.
	The purpose of the initialization is to ensure that exactly one simulated TM head spawns at the beginning and no further ones can spawn later on.
	
	After the initialization, we can start simulating TM transitions.
	\begin{equation}
		\begin{aligned}
			\big((\gamma,n,q,H),\mathcal{N}_v^C\big) & \mapsto (\gamma',n,q',H_{d}) & \text{for } & (q',\gamma',d) = \delta(q,\gamma) \\
			\big((\gamma,n,q,H_d),\mathcal{N}_v^C\big) & \mapsto (\gamma,n) & \text{for } & \mathcal{N}_v^C(*,(n+d)\bmod{3}) \geq 1 \\
			\big((\gamma,n),\mathcal{N}_v^C\big) & \mapsto (\gamma,n,q,H) & \text{for } & \mathcal{N}_v^C(*,(n-d)\bmod{3},q,H_d) \geq 1
		\end{aligned}\labelledtag{H-tra}
	\end{equation}
	Intuitively, a simulated TM transition consists of three actions, performed in two steps: First, the next TM transition is initiated by the agent $u$ that simulates the TM head by updating its tape symbol and the TM state, and indicating the intended moving direction of the TM head.
	Then, $u$ makes sure that a neighbour $w$ in the intended moving direction exists before dropping the TM state and TM head state, returning to a $Q'$-state.
	Simultaneously, $w$ detects the intention and adopts the TM state from $u$ and the TM head state $H$, simulating the TM head.
	
	If, however, there is no next TM transition to be executed, i.e., the TM halts, or there is no agent in the intended moving direction, i.e., the TM uses more tape cells than the graph represents, the agent simulating the TM head transitions to the accepting or rejecting state, respectively:
	\begin{align}
		\big((\gamma,n,q,H),\mathcal{N}_v^C\big) & \mapsto \checked & \text{for } & \delta(q,\gamma) = \bot, \labelledtag{H-halt} \\
		\big((\gamma,n,q,H_d),\mathcal{N}_v^C\big) & \mapsto \bot & \text{for } & \mathcal{N}_v^C(*,(n+d)\bmod{3}) = 0. \labelledtag{H-overflow}
		\intertext{Note that there is only one simulated TM head, which gets replaced in either of the two transitions above.
			Therefore, only either $\checked$ states or $\bot$ states can occur in a run, never both.
			The following propagation transition is thus well-defined and complies with halting acceptance ($\mathtt{a}$).
			We define for $r \in \{\checked,\bot\}$}
		\big(q,\mathcal{N}_v^C\big) & \mapsto r & \text{for } & \mathcal{N}_v^C(r) \geq 1. \labelledtag{H-prop}
	\end{align}
	
	All other transitions are silent.
\end{construction}

To prove the correctness of our simulation, we first formally introduce TM configurations and transitions.
While we assume that the reader is already familiar with these conceptually, our notation differs from the conventional TM model.
This specific notation will play a crucial role in the subsequent lemma.

A \emph{configuration} $C$ of a TM $T = (Q,q_0,F,\Gamma,\Sigma,\square,\delta)$ is a triple $(q,\theta,p) \in Q \times \Gamma^+ \times \mathbb{N}$, representing the Turing machine's state, the tape word, i.e., the word of symbols on the tape, and the ($0$-indexed) position of the Turing machine's head in the tape word.
Note that the tape word can by definition not be empty as there always has to be a symbol $\theta_p$ for the Turing machine's head to point at, even if it is only a blank symbol $\square$.
For two configurations $(q,\theta,p),(q',\theta',p') \in Q \times \Gamma^+ \times \mathbb{N}$ with $\delta(q,\theta_p) = (q',\gamma,d)$, we denote by $(q,\theta,p) \to_T (q',\theta',p')$ that
\begin{equation}
	\begin{aligned}
		\theta' & = \begin{cases}
			\theta_{0...p-1} \gamma \theta_{p+1...|\theta|-1} & \text{for } p+d \neq |\theta|, \\
			\theta_{0...p-1} \gamma \theta_{p+1...|\theta|-1} \square & \text{for } p+d = |\theta|,
		\end{cases} \\
		p' & = p+d.
	\end{aligned}\labelledtag{TM}
\end{equation}
Note that the tape word can only ever get longer and in particular does not omit blank symbols from the configuration.
Further note what happens in the edge cases of the head's position:
If it exceeds the end of the tape word, a blank symbol is appended to the tape word.
The symbol at the head's position $\theta_p$ is therefore always well-defined.
At the other end, the beginning of the tape word can never be exceeded since $p' \in \mathbb{N}$, realizing the aforementioned left-boundedness of the TM tape.

With this notation the halting problem may be phrased as follows:
Decide whether there are $m \in \mathbb{N}, (q,\theta,p) \in Q \times \Gamma^+ \times \mathbb{N}$, such that $(q_0,\square,0) \to_T^m (q,\theta,p)$ and $\delta(q,\theta_p) = \bot$.

Lastly, we formalize how a configuration of $T$ is represented as a configuration of the distributed machine $M^H(T)$ on an NLG $G = (V = \{v_0,\dots,v_{n-1}\},E,\lambda) \in \mathsf{NLG}$ of length $n \in \mathbb{N} \setminus \{0\}$.
For $(q,\theta,p) \in Q \times \Gamma^+ \times \mathbb{N}, i \in [n-1]$, we define
$$C_{(q,\theta,p)}(v_i) \coloneq
\begin{cases}
	({(\theta\square^\omega)}_i,i\bmod{3}) & \text{for } i \neq p, \\
	({(\theta\square^\omega)}_p,p\bmod{3},q,H) & \text{for } i = p.
\end{cases}$$

\begin{lemma}\label{lem:simulation}
	Let $T = (Q,q_0,F,\Gamma,\Sigma,\square,\delta)$ be a Turing machine and $M^H(T)$ be \Cref{cstr:tm_head} for $T$.
	Further, let $(q,\theta,p) \in Q \times \Gamma^+ \times \mathbb{N}$ be a configuration of $T$ and $\rho = {(C_i)}_{i \in \mathbb{N}}$ be the (unique) run of $M^H(T)$ on a numbered linear graph $G \in \mathsf{NLG}$ of length $n \geq |\theta|$.
	Then, for any $m \in \mathbb{N}$, we have
	$$(q_0,\square,0) \to_T^m (q,\theta,p) \implies C_{2m+1} = C_{(q,\theta,p)}.$$
\end{lemma}

Note that it is important here that we do not omit blank symbols from the tape word $\theta$.
$|\theta|$ thereby encodes exactly how many tape cells have been visited by the TM head in the transitions leading up to the configuration $(q,\theta,p)$.
This information is crucial to determine the minimal graph size required to simulate the TM's run until the given configuration.

\begin{proof}
	We proceed by induction on $m$.
	
	We start with the base case for $m=0$.
	On the left-hand side, we get $(q_0,\square,0) \to_T^0 (q_0,\square,0)$.
	For $\rho$, we get $C_0(v_i) = \delta_0^H(\lambda(v_i)) = \delta_0^H(i\bmod{3}) = (\circ,\square,i\bmod{3})$ for all $i \in [n-1]$ according to \hyperref[itm:nlg_numbering]{(L2)}.
	In the first transition, all agents perform the respective initialization transitions.
	Since $G$ is numbered correctly, the only agent to perform \ref{H-origin} is at the origin node $v_0$, while all other agents $v_i,i \geq 1$ perform \ref{H-init}.
	This yields
	\begin{flalign*}
		&&C_1(v_0) & = (\square,0,q_0,H) = ({(\square\square^\omega)}_0,0,q_0,H) = C_{(q_0,\square,0)}(v_0) \text{, and}&& \\
		&&C_1(v_i) & = (\square,i\bmod{3}) = ({(\square\square^\omega)}_i,i\bmod{3}) = C_{(q_0,\square,0)}(v_i)&&
	\end{flalign*}
	for $i \geq 1$, concluding the base case.
	
	For the induction step, we assume that the claim holds for $m-1$ and prove
	$$(q_0,\square,0) \to_T^m (q',\theta',p') \implies C_{2m+1} = C_{(q',\theta',p')}.$$
	The left-hand side is equivalent to the existence of $(q,\theta,p) \in Q \times \Gamma^+ \times \mathbb{N}$, such that
	$$(q_0,\square,0) \to_T^{m-1} (q,\theta,p) \to_T (q',\theta',p')$$
	with $\delta(q,\theta_p) = (q',\gamma,d)$, $p' = p+d$ and the dependency between $\theta$ and $\theta'$ as defined in \ref{TM}.
	We can apply the induction hypothesis to the first relation to get $C_{2m-1} = C_{(q,\theta,p)}.$
	When $M^H(T)$ transitions from $C_{(q,\theta,p)}$ to $C_{2m}$, the agent $v_p$ performs the first \ref{H-tra} transition from $({(\theta\square^\omega)}_p,p\bmod{3},q,H) = (\theta_p,p\bmod{3},q,H)$ to $(\gamma,p\bmod{3},q',H_d)$.
	All other agents perform silent transitions.
	Note that $v_p$ will always have a neighbour $v_{p+d}=v_{p'}$ because $p' \geq 0$ and $p' \leq |\theta'|-1 \leq n-1$.
	In the next step, $v_p$ therefore detects its neighbour $v_{p+d}$ with numbering $(p+d)\bmod{3}$ and executes the second \ref{H-tra} transition to arrive at the state $(\gamma,p\bmod{3})$.
	$v_{p'}$ in turn, detects its neighbour $v_p=v_{p'-d}$ with numbering $(p'-d)\bmod{3}$ and a TM-state $q'$, indicating the intended moving direction $d$.
	Thus, applying the third \ref{H-tra} transition, $v_{p'}$ transitions from $({(\theta\square^\omega)}_{p'},p'\bmod{3})$ to $({(\theta\square^\omega)}_{p'},p'\bmod{3},q',H)$.
	All other agents again transition silently.
	Using that $p' = p+d \neq p$, we can conclude
	\begin{flalign*}
		&&C_{2m+1}(v_p) & = (\gamma,p\bmod{3})&& \\
		&&& = ({(\theta_{0...p-1} \gamma \theta_{p+1...|\theta|-1} \square^\omega)}_p,p\bmod{3})&& \\
		&&& = ({(\theta'\square^\omega)}_p,p\bmod{3}) = C_{(q',\theta',p')}(v_p),&&
		\intertext{and}
		&&C_{2m+1}(v_{p'}) & = ({(\theta\square^\omega)}_{p'},p'\bmod{3},q',H)&& \\
		&&& = ({(\theta_{0...p-1} \gamma \theta_{p+1...|\theta|-1} \square^\omega)}_{p'},p'\bmod{3},q',H)&& \\
		&&& = ({(\theta'\square^\omega)}_{p'},p'\bmod{3},q',H) = C_{(q',\theta',p')}(v_{p'}).&&
		\intertext{All other agents only performed silent transitions, so for $i \notin \{p,p'\}$, we get}
		&&C_{2m+1}(v_i) & = ({(\theta\square^\omega)}_i,i\bmod{3})&& \\
		&&& = ({(\theta_{0...p-1} \gamma \theta_{p+1...|\theta|-1} \square^\omega)}_i,i\bmod{3})&& \\
		&&& = ({(\theta'\square^\omega)}_i,i\bmod{3}) = C_{(q',\theta',p')}(v_i),&&
	\end{flalign*}
	completing the induction step.
\end{proof}

This shows that the simulation is indeed correct, as long as the NLG is sufficiently large and the TM has not halted.
Next, we analyse what happens when either of these prerequisites are violated.

\begin{lemma}\label{lem:case_halt/overflow}
	Let $T = (Q,q_0,F,\Gamma,\Sigma,\square,\delta)$ be a Turing machine and $M^H(T)$ be \Cref{cstr:tm_head} for $T$.
	Further, let $\rho_n$ be the (unique) run of $M^H(T)$ on a numbered linear graph $G \in \mathsf{NLG}$ of length $n \in \mathbb{N} \setminus \{0\}$.
	\begin{bracketenumerate}
		\item If $T$ on blank tape visits every tape cell, then $\rho_n$ is rejecting for all $n \in \mathbb{N}$.\label{itm:h/o_1}
		\item If $T$ halts on blank tape, then there exists a threshold $n_0 \in \mathbb{N}$, such that $\rho_n$ is accepting for all $n \geq n_0$, and rejecting otherwise.\label{itm:h/o_2}
	\end{bracketenumerate}
\end{lemma}
\begin{proof}
	First, we assume that $T$ visits every tape cell.
	Consider the run $\rho_n = {(C_i)}_{i \in \mathbb{N}}$ on $G = (V = \{v_0,\dots,v_{n-1}\},E,\lambda)$.
	Since $T$ visits every tape cell, it will at some point visit tape cell $n$ (recall $0$-indexing).
	Let $m_0 \in \mathbb{N}$, such that $T$ visits tape cell $n$ for the first time in transition $m_0$.
	In particular, $T$ has at most visited the first $n$ tape cells $0$ to $n-1$ in the first $m_0-1$ transitions.
	Therefore, there are $(q,\theta,p),(q',\theta',p') \in Q \times \Gamma^+ \times \mathbb{N}$, such that $(q_0,\square,0) \to_T^{m_0-1} (q,\theta,p) \to_T (q',\theta',p')$ with $|\theta'| = n+1$ and $|\theta| \leq n$.
	From this, we can reconstruct the transition $\delta(q,\theta_p) = (q',\theta'_p,d)$.
	Evidently, the length of the tape word increases in transition $m_0$, which only happens if $p+d = |\theta|$ where we append a blank symbol to the tape word.
	We can hence deduce $d = {+1}$, since $p \leq |\theta|-1$, as well as $|\theta|+1 = |\theta'| = n+1$, and thus $p = |\theta|-1 = n-1$.
	As $n = |\theta|$, we can employ \Cref{lem:simulation}, to get that $C_{2m_0-1} = C_{(q,\theta,p)} = C_{(q,\theta,n-1)}$ and therefore
	$$C_{2m_0-1}(v_{n-1}) = ((\theta\square^\omega)_{n-1},(n-1)\bmod{3},q,H) = (\theta_p,p\bmod{3},q,H).$$
	Next, the agent $v_{n-1}$ performs the first \ref{H-tra} transition to the state $(\theta'_{p},p\bmod{3},q',H_{+1})$.
	Since $v_{n-1}$ only has one neighbour $v_{n-2}=v_{p-1}$, there is no neighbour with numbering $(p+1)\bmod{3}$.
	Therefore, the second \ref{H-tra} transition does not apply, but rather \ref{H-overflow}, producing a rejecting state $\bot$.
	The rejecting state propagates through the graph with \ref{H-prop}, rejecting $G$.
	Note that, as mentioned in \Cref{cstr:tm_head}, both accepting and rejecting states can only be produced by replacing the TM head state.
	Since there is at most one agent in a TM head state at any time, either agents with $\checked$ states or agents with $\bot$ states can exist, never both.
	Consequently, the propagation does indeed follow through and $\rho_n$ is rejecting.
	\medskip
	
	Now, we assume that $T$ halts on blank tape, i.e., there are $m \in \mathbb{N}, (q,\theta,p) \in Q \times \Gamma^+ \times \mathbb{N}$, such that $(q_0,\square,0) \to_{T}^m (q,\theta,p)$ and $\delta(q,\theta_p) = \bot$.	
	Consider the run $\rho_n = {(C_i)}_{i \in \mathbb{N}}$ on $G = (V = \{v_0,\dots,v_{n-1}\},E,\lambda)$ with length $n \geq |\theta|$.
	According to \Cref{lem:simulation}, we have $C_{2m+1} = C_{(q,\theta,p)}$ and therefore
	$$C_{2m+1}(v_p) = ((\theta\square^\omega)_p,p\bmod{3},q,H) = (\theta_p,p\bmod{3},q,H).$$
	Since $\delta(q,\theta_p) = \bot$, the agent $v_p$ will perform the transition \ref{H-halt} next, producing an accepting state $\checked$.
	The accepting state propagates through the graph with \ref{H-prop} (which works out as argued above), making $\rho_n$ accepting.
	
	Lastly, consider the run $\rho_n$ for $n < |\theta|$.
	Evidently, $T$ visits the $|\theta|$ tape cells $0$ to $|\theta|-1$, in particular tape cell $n \leq |\theta|-1$.
	The rejection of $G$ follows as in the proof of claim \ref{itm:h/o_1}.
	
	Therefore, the claimed $n_0$ is exactly $|\theta|$.
\end{proof}

Notably, the above lemma is not exhaustive, since it does not deal with the case of a TM that does not halt, but also does not visit every tape cell.
This case can actually occur if the TM gets stuck in a cycle between a finite set of configurations.
Then, we would want $\rho_n$ to be rejecting for all $n \in \mathbb{N} \setminus \{0\}$, which is not the case in $M^H(T)$.
To solve this issue, we will actually simulate a different TM derived from $T$ where this neglected case cannot occur.

\begin{lemma}\label{lem:always_halt/overflow}
	Let $T$ be a Turing machine.
	There exists a Turing machine $T_\infty$ that halts if{}f $T$ halts and otherwise visits every tape cell.
\end{lemma}
\begin{proof}
	We describe the behaviour of such a Turing machine $T_\infty$.
	$T_\infty$ will perform exactly the transitions of $T$ and additionally the following subroutine after each $T$-transition:
	When performing a $T$-transition, a visited-marker is added to the symbol in the tape cell before moving in the specified direction.
	The newly reached tape cell is marked as current.
	The head then moves to the right as long as tape cells are marked.
	When the first unmarked tape cell is reached, it gets marked as visited.
	Lastly, the head moves back to the current-marked tape cell and removes the current-marker.
	visited-markers are never removed.
	From there, $T_\infty$ performs the next $T$-transition and the subroutine starts anew.
	
	It is easy to see that after $T_\infty$ performs a $T$-transition and the subroutine, each tape cell of $T_\infty$ carries the same symbol (besides the markers) as it would in $T$ and the head is on the same tape cell as it would be in $T$.
	So given that $T$ halts after $m \in \mathbb{N}$ transitions, $T_\infty$ halts after executing a $T$-transition and the subroutine $m$ times.
	If $T$ does not halt, it performs infinitely many transitions and therefore $T_\infty$ performs the subroutine infinitely many times.
	Since $T_\infty$ visits the first previously unvisited tape cell in each execution of the subroutine, it eventually visits every tape cell as the tape is bounded to the left.
\end{proof}

Combining \Cref{lem:case_halt/overflow,lem:always_halt/overflow} immediately yields that $A^H \coloneq (M^H(T_\infty), \Sigma)$, where $M^H(T_\infty)$ is \Cref{cstr:tm_head} for $T_\infty$ and $\Sigma$ a matching scheduler, is the $\mathsf{NLG}$-distributed automaton in question for \Cref{cor:always_acc/rej}.

%% file: chapters/appendix_theorem_1.tex
\reduction*

\begin{proof}[Proof for $\mathtt{DAf}$ and $\mathtt{DAF}$]
	Given a TM $T$, let $A^L$ be the $\mathtt{DA\texttt{\$}f}$-distributed automaton from \Cref{lem:nlg}, and $A^H$ the $\mathsf{NLG}$-distributed automaton obtained from applying \Cref{cor:always_acc/rej} to $T$ and lifting it to the class $\mathtt{DA\texttt{\$}f}$, which is possible by the hierarchy of expressive power.
	As $L(A^L) = \mathsf{NLG}$, \Cref{lem:intersection} yields a $\mathtt{DA\texttt{\$}f}$-distributed automaton $A^T$ with $L(A^T) = L(A^L) \cap L(A^H)$.
	We show that $L(A^T) \neq \emptyset$ if{}f $T$ halts on blank tape.
	
	If $L(A^T) \neq \emptyset$, let $G \in L(A^T) = L(A^L) \cap L(A^H)$.
	As $G \in L(A^L) = \mathsf{NLG}$, $G$ is an NLG.
	As $G \in L(A^H)$ for $G$ an NLG, $T$ halts on blank tape by \Cref{cor:always_acc/rej}.
	
	Conversely, if $T$ halts on blank tape, then so does $T_\infty$ from \Cref{lem:always_halt/overflow}.
	This happens after a finite number of steps, visiting a finite number of tape cells $n_0 \in \mathbb{N} \setminus \{0\}$.
	By \Cref{lem:nlg,cor:always_acc/rej}, $A^L$ and $A^H$ both accept the NLG $G \in \mathsf{NLG}$ of length $n_0$.
	We conclude $G \in L(A^L) \cap L(A^H) = L(A^T)$, so $L(A^T) \neq \emptyset$.
	
	By \Cref{prop:sync} and the hierarchy of expressive power, this shows the claim for the classes $\mathtt{DAf}$ and $\mathtt{DAF}$.
\end{proof}

%% file: chapters/appendix_nqlg.tex
For technical reasons we first prove \Cref{prop:nqlg_eq_nlg} and then \Cref{lem:nqlg}.

\subsubsection{Proof of \Cref{prop:nqlg_eq_nlg}}\label{app:nqlg_eq_nlg}

\propNqlgEqNlg*

First, we prove a lemma establishing a few facts about distances and NQLGs to help with the proof of \Cref{prop:nqlg_eq_nlg}.

\begin{lemma}\label{lem:nqlg_properties}
	Let $G = (V,E,\lambda)$ be a graph where $O \subseteq V$ as defined in \Cref{def:nqlg} is non-empty.
	Then,
	\begin{bracketenumerate}
		\item $\forall\{v,w\} \in E \colon {\lvert\dist(O,w)-\dist(O,v)\rvert\leq1}$; if $G \in \mathsf{NQLG}$, we even get equality for all $\{v,w\} \in E$,\label{itm:nqlg_prop_1}
		\item for all $v \in V$\label{itm:nqlg_prop_2}$$\big(\exists \{u,v\} \in E \colon {\dist(O,u) = \dist(O,v)-1}\big) \iff \dist(O,v)>0 \text{, and}$$
		\item if $G \in \mathsf{NQLG}$ is a numbered quasi-linear graph of length $n$, we get $\forall v \in V \colon {\dist(O,v)<n}$.\label{itm:nqlg_prop_3}
	\end{bracketenumerate}
\end{lemma}
\begin{proof}
	For claim \ref{itm:nqlg_prop_1}, assume there was an edge $\{v,w\} \in E$ with $\lvert\dist(O,w)-\dist(O,v)\rvert>1$.
	Without loss of generality, let $\dist(O,w) > \dist(O,v)$ and therefore $\dist(O,w) > \dist(O,v)+1$.
	This means that there is a path $(e_1,...,e_{\dist(O,v)})$ of length $\dist(O,v)$ from a node $o \in O$ to $v$.
	Obviously, $(e_1,...,e_{\dist(O,v)},\{v,w\})$ would then be a path of length $\dist(O,v)+1$ from $o$ to $w$, implying $\dist(O,w) \leq \dist(O,v)+1$, a contradiction.
	We thus know $\lvert\dist(O,w)-\dist(O,v)\rvert\leq1$ for all $\{v,w\} \in E$.
	In the case $G \in \mathsf{NQLG}$, if there was an edge $\{v,w\} \in E$ with $\lvert\dist(O,w)-\dist(O,v)\rvert=0$, we would have $\lambda(v) = \dist(O,v)\bmod{3} = \dist(O,w)\bmod{3} = \lambda(w)$, contracting \hyperref[itm:nqlg_neighbours]{(QL2)}.
	
	For claim \ref{itm:nqlg_prop_2}, we first consider a node $v \in V$ with $\dist(O,v)=0$.
	Obviously, an edge $\{u,v\} \in E$ with $\dist(O,u) = \dist(O,v)-1=-1$ does not exist.
	Now, consider a node $v \in V$ with $\dist(O,v)>0$.
	Again, this means that there is a path $(e_1,...,e_{\dist(O,v)})$ of length $\dist(O,v)$ from a node $o \in O$ to $v$.
	Since $\dist(O,v)>0$, the edge $e_{\dist(O,v)} = \{u,v\}$ exists.
	Therefore, $(e_1,...,e_{\dist(O,v)-1})$ is a path of length $\dist(O,v)-1$ from $o$ to $u$, implying $\dist(O,u) \leq \dist(O,v)-1$.
	By claim \ref{itm:nqlg_prop_1}, this implies $\dist(O,u) = \dist(O,v)-1$.
	
	Lastly, for claim \ref{itm:nqlg_prop_3}, let $G \in \mathsf{NQLG}$ be an NQLG of length $n$ and let $v \in V$.
	Assume $\dist(O,v) \geq n$.
	If $\dist(O,v)=n$, define $v_n=v$.
	Otherwise, by iteratively applying claim \ref{itm:nqlg_prop_2}, we obtain a node $v_k \in V$ with $\dist(O,v_k)=k$ for all $k \leq \dist(O,v)$.
	In both cases, applying claim \ref{itm:nqlg_prop_2} to $v_n$ one more time, yields an edge $\{u,v_n\} \in E$ with $\dist(O,u) = \dist(O,v_n)-1 = n-1$ and $\dist(O,v_n) = \dist(O,u)+1$, contradicting \hyperref[itm:nqlg_length]{(QL3)}.
\end{proof}

We can now proceed to prove \Cref{prop:nqlg_eq_nlg}.

\begin{proof}[Proof of \Cref{prop:nqlg_eq_nlg}]
	First, note that the expression in the claim is always well-defined: Since $\tilde{G}$ and $G$ have the same length $n$, for all $\tilde{v} \in \tilde{V}$, \Cref{lem:nqlg_properties} \hyperref[itm:nqlg_prop_3]{(3)} states $\dist(O,\tilde{v})<n$ and thus the node $v_{\dist(O,\tilde{v})} \in G$ exists.
	
	We will now prove the claim using induction on $i$.
	
	In the base case $i=0$, we get for every $\tilde{v} \in \tilde{V}$
	$$\tilde{C}_0(\tilde{v}) = \delta_0(\tilde{\lambda}(\tilde{v})) = \delta_0(\dist(O,\tilde{v})\bmod{3}) = \delta_0(\lambda(v_{\dist(O,\tilde{v})})) = C_0(v_{\dist(O,\tilde{v})}).$$
	
	For the induction step, we assume the claim to be true for an arbitrary but fixed $i \in \mathbb{N}$ and all $\tilde{v} \in \tilde{V}$.
	Consider $\tilde{v} \in \tilde{V}$.
	If $0<\dist(O,\tilde{v}) < n-1$, there are edges $\{\tilde{u},\tilde{v}\},\{\tilde{v},\tilde{w}\} \in \tilde{E}$ with $\dist(O,\tilde{u})+1 = \dist(O,\tilde{v}) = \dist(O,\tilde{w})-1$ according to \Cref{lem:nqlg_properties} \hyperref[itm:nqlg_prop_2]{(2)} and \hyperref[itm:nqlg_length]{(QL3)}.
	For all neighbours $\tilde{v}' \in \tilde{V}$ of $\tilde{v}$, \Cref{lem:nqlg_properties} \hyperref[itm:nqlg_prop_1]{(1)} yields
	$$\dist(O,\tilde{v}') \in \{\dist(O,\tilde{v})-1,\dist(O,\tilde{v})+1\} = \{\dist(O,\tilde{u}),\dist(O,\tilde{w})\}.$$
	and thus, by the induction hypothesis,
	$$\tilde{C}_i(\tilde{v}') = C_i(v_{\dist(O,\tilde{v}')}) \in \{C_i(v_{\dist(O,\tilde{u})}),C_i(v_{\dist(O,\tilde{w})})\}.$$
	Since $0<\dist(O,\tilde{v}) < n-1$, $v_{\dist(O,\tilde{v})}$ has exactly the neighbours $v_{\dist(O,\tilde{v})-1}$ and $v_{\dist(O,\tilde{v})+1}$.
	Because the counting bound in the $\mathtt{dA\texttt{\$}f}$-distributed machine $M$ is $\beta = 1$ ($\mathtt{d}$), it follows for all $q \in Q$ that
	$$\mathcal{N}_{\tilde{v}}^{\tilde{C}_i}(q) = \left.\begin{cases}
		1 & \text{if } C_i(v_{\dist(O,\tilde{u})}) = q = C_i(v_{\dist(O,\tilde{v})-1}), \\
		1 & \text{if } C_i(v_{\dist(O,\tilde{w})}) = q = C_i(v_{\dist(O,\tilde{v})+1}), \\
		0 & \text{else} \\
	\end{cases}\right\} = \mathcal{N}_{v_{\dist(O,\tilde{v})}}^{C_i}(q).$$
	The same equality holds if $\dist(O,\tilde{v}) \in \{0,n-1\}$.
	The argument is completely analogous with the only difference that both $\tilde{v}$ and $v_{\dist(O,\tilde{v})}$ only have successors (predecessors).
	In the next transition, $\tilde{v}$ and $v_{\dist(O,\tilde{v})}$ are each selected in the respective graphs as $M$ employs synchronous scheduling ($\texttt{\$}$).
	Using the induction hypothesis and the equality above, we get
	$$\tilde{C}_{i+1}(\tilde{v}) = \delta\big(\tilde{C}_{i}(\tilde{v}),\mathcal{N}_{\tilde{v}}^{\tilde{C}_i}\big) = \delta\big(C_{i}(v_{\dist(O,\tilde{v})}),\mathcal{N}_{v_{\dist(O,\tilde{v})}}^{C_i}\big) = C_{i+1}(v_{\dist(O,\tilde{v})})$$
	as claimed.	
	Since $\tilde{v} \in \tilde{V}$ was arbitrary, this concludes the induction step.
\end{proof}

\subsubsection{Proof of \Cref{lem:nqlg}}\label{app:nqlg}

\lemNqlg*

We formalize the construction described in the proof sketch of \Cref{lem:nqlg}:

\begin{construction}[Deciding Numbered Quasi-Linear Graphs]\label{cstr:nqlg}
	We construct a $\mathtt{dA\texttt{\$}f}$-distributed machine $M^L = (Q^L,\delta_0^L,\delta^L,Y^L,N^L)$ with:
	\begin{itemize}
		\item The labelling alphabet is $\Lambda^L = \mathbb{Z}_3$.
		\item $Q^L = (\Lambda^L \times \{0,1,2\}) \cup \{\bot\}.$
		Compared to \Cref{cstr:nlg}, the guess is replaced by one of three detection stages: $0$ -- the initial stage, $1$ -- origin set detected, similar to guess $1$, or $2$ -- origin set and end of graph detected.
		\item $\delta_0^L \colon {\Lambda^L \to Q^L, l \mapsto (l,0)}.$
		The numbering is again taken directly from the labelling, and all agents start in the initial stage $0$.
		\item The accepting states are $Y^L = \mathbb{Z}_3 \times \{2\}$, i.e., the states where the agent detected the origin set and the end of the graph, while the rejecting states are $N^L = (\mathbb{Z}_3 \times \{0\}) \cup \{\bot\}$, i.e., the states where the agent remained in the initial stage and the error state.
	\end{itemize}
	For an agent $v \in V$ and a configuration $C \colon {V \to Q}$, we define the transition function $\delta^L$:
	\begin{equation}
		\begin{aligned}
			\big((n,s),\mathcal{N}_v^C\big) & \mapsto \bot & \text{for } & \mathcal{N}_v^C(n,*) \geq 1 \\
			\big(q,\mathcal{N}_v^C\big) & \mapsto \bot & \text{for } & \mathcal{N}_v^C(\bot) \geq 1.
		\end{aligned}\labelledtag{QL-error}
	\end{equation}
	An agent will transition to the error state if it detects a neighbour with the same numbering as itself, as this violates \hyperref[itm:nqlg_neighbours]{(QL2)}.
	Note that in contrast to \Cref{cstr:nlg}, multiple neighbours with the same numbering do not produce an error state, as $M^L$ only has existence detection ($\mathtt{d}$).
	Furthermore, if any neighbour of $v$ already is in the error state, $v$'s next state is $\bot$, so that errors propagate through the graph unconditionally.
	For all the following transitions, we assume that, in addition to the specified conditions being fulfilled, the conditions for the \ref{QL-error} transitions are not met:
	\begin{align}
		\big((0,0),\mathcal{N}_v^C\big) & \mapsto (0,1) & \text{for } & \mathcal{N}_v^C(2,*) = 0 \labelledtag{QL-origin}
	\end{align}
	If the numbering indicates that $v$ does not have predecessors, $v$ is an origin node and initiates stage $1$.
	\begin{equation}
		\begin{aligned}
			\big((n,0),\mathcal{N}_v^C\big) & \mapsto (n,1) & \text{for } & \mathcal{N}_v^C((n-1)\bmod{3},1) \geq 1\ \wedge \\
			&&& \mathcal{N}_v^C((n+1)\bmod{3},1) = 0 \\
			\big((n,0),\mathcal{N}_v^C\big) & \mapsto \bot  & \text{for } & \mathcal{N}_v^C((n+1)\bmod{3},1) \geq 1 \\
		\end{aligned} \labelledtag{QL-stage1}
	\end{equation}
	If $v$ is in stage $0$ and detects that one of its predecessors reached stage $1$, it will advance to stage $1$ as well, except if the second transition applies.
	Since the schedule is synchronous ($\texttt{\$}$), all agents at a given distance from the origin set $O$, should reach stage $1$ at the same time.
	So, if a presumed successor of $v$ reaches stage $1$ before $v$ itself, this means that the numbering is incorrect and that neighbour should actually be a predecessor of $v$.
	$v$ detects this fault and transitions to $\bot$.
	\begin{align}
		\big((n,1),\mathcal{N}_v^C\big) & \mapsto (n,2) & \text{for } & \mathcal{N}_v^C((n+1)\bmod{3},*) = 0, \labelledtag{QL-limit}
	\end{align}
	If $v$ does not have any successors, $v$ should, by \hyperref[itm:nqlg_length]{(QL3)}, be one of the agents with the maximum distance from $O$ of $n-1$.
	The end of the graph has therefore been reached and $v$ transitions from stage $1$ to stage $2$.
	Nodes like $v$ are called \textit{limit nodes}.
	\begin{equation}
		\begin{aligned}
			\big((n,1),\mathcal{N}_v^C\big) & \mapsto (n,2) & \text{for } & \mathcal{N}_v^C((n+1)\bmod{3},2) \geq 1\ \wedge \\
			&&& \mathcal{N}_v^C((n-1)\bmod{3},2) = 0, \\
			\big((n,1),\mathcal{N}_v^C\big) & \mapsto \bot  & \text{for } & \mathcal{N}_v^C((n-1)\bmod{3},2) \geq 1,
		\end{aligned}\labelledtag{QL-stage2}
	\end{equation}
	Lastly, if $v$ is in stage $1$ and detects that one of its successors reached stage $2$, it will advance to stage $2$ as well, except if the second transition applies.
	Due to the unambiguous length condition \hyperref[itm:nqlg_length]{(QL3)} and synchronous scheduling ($\texttt{\$}$), $v$ should detect the end of the graph before its predecessors.
	If this is not the case, it indicates that there is a shorter path from that predecessor to a limit node than through $v$, contradicting \hyperref[itm:nqlg_length]{(QL3)}.
	$v$ thus transitions to $\bot$.
	
	All other transitions are silent.
\end{construction}

\begin{proof}[Proof of \Cref{lem:nqlg}]
	Let $A^L = (M^L, \Sigma)$, where $M^L$ is \Cref{cstr:nqlg} and $\Sigma$ a matching scheduler.
	First, we establish a fact about the execution of $A^L$.
	Assuming that $O$, as defined in \Cref{def:nqlg}, is non-empty and no error states occur, we will show inductively that for all $d \in \mathbb{N}$, in transition $d$ (we number the transitions starting with $0$), exactly the agents $v \in V$ with $\dist(O,v)=d$ reach stage $1$.
	
	In transition $0$, exactly the agents in $O$ employ \ref{QL-origin} to reach stage $1$.
	Since no error state occurs, all other agents stay in stage $0$.
	Assume the above claim for all $d'$ up to an arbitrary but fixed $d \in \mathbb{N}$.
	Let $v \in V$ with $\dist(O,v) = d+1$.
	By \Cref{lem:nqlg_properties} \hyperref[itm:nqlg_prop_2]{(2)}, $v$ has a neighbour $u$ with $\dist(O,u)=d$ and by assumption, $u$ reaches stage $1$ in transition $d$.
	Therefore, in transition $d+1$, $v$ will transition to stage $1$ or to $\bot$ by \ref{QL-stage1}.
	All agents $u \in V$ with $\dist(O,u) \leq d$ have by assumption already reached stage $1$.
	An agent $w \in V$ with $\dist(O,w) \geq d+2$ only has neighbours $v \in V$ with $\dist(O,v) \geq d+1$ according to \Cref{lem:nqlg_properties} \hyperref[itm:nqlg_prop_1]{(1)}.
	Therefore, by assumption, $w$ does not have neighbours that had already reached stage $1$.
	This shows that no other agents reach stage $1$ in transition $d+1$.
	Inductively, this proves the claim above.
	\medskip
	
	Now, we assume that $G = (V,E,\lambda)$ is, in fact, an NQLG of length $n \in \mathbb{N} \setminus \{0\}$.
	Because of \hyperref[itm:nqlg_neighbours]{(QL2)}, the first \ref{QL-error} transition will never take effect.
	Since $G \in \mathsf{NQLG}$, all presumed successors $w$ of an agent $v \in V$ fulfil $\dist(O,w)=\dist(O,v)+1$ by \Cref{lem:nqlg_properties} \hyperref[itm:nqlg_prop_1]{(1)} and are therefore still in stage $0$ in the above induction step.
	This implies that the case where \ref{QL-stage1} produces a $\bot$ state cannot occur.
	Thus, all agents will reach stage $1$ without any error states occurring.
	By \Cref{lem:nqlg_properties} \hyperref[itm:nqlg_prop_3]{(3)} and the induction, this happens in exactly $n$ transitions.
	Then, in transition $n$, all agents at limit nodes, which by \hyperref[itm:nqlg_length]{(QL3)} are those at distance $n-1$ from the origin set, transition from stage $1$ to $2$ by \ref{QL-limit}.
	Analogous to the argument for stage $1$, for all $d \in [n-1]$, in transition $n+d$, exactly the agents $v \in V$ with $\dist(O,v) = n-1-d$ reach stage $2$ by \ref{QL-stage2}.
	Again, no errors occur.
	After a total of $2n$ transitions, all agents reached stage $2$, accepting $G$.
	\medskip
	
	Now, we assume that $G = (V,E,\lambda)$ is not an NQLG, that is, $O$ is empty or \hyperref[itm:nqlg_numbering]{(QL1)}, \hyperref[itm:nqlg_neighbours]{(QL2)}, or \hyperref[itm:nqlg_length]{(QL3)} are violated.
	
	If $O$ is empty, then there are no agents that can perform the transition \ref{QL-origin}.
	Therefore, all agents that do not produce error states will stay in the initial stage $0$ forever, rejecting $G$ by stable consensus.
	
	If there are neighbouring nodes with the same numbering (violating \hyperref[itm:nqlg_neighbours]{(QL2)}), the corresponding agents will perform the first \ref{QL-error} transition, producing a $\bot$ state.
	This will propagate through the graph with the second \ref{QL-error} transition, rejecting $G$.
	
	If no such neighbouring nodes exist, but a node where the numbering does not match its distance to $O$ modulo $3$ (violating \hyperref[itm:nqlg_numbering]{(QL1)}), consider a node $v \in V$ with $\lambda(v) \neq \dist(O,v)\bmod{3}$ where $\dist(O,v)$ is minimal.
	Note that the nodes in $O$ have the correct numbering $0$, so $v \notin O$ and hence $\dist(O,v)>0$.
	By \Cref{lem:nqlg_properties} \hyperref[itm:nqlg_prop_2]{(2)}, there is a neighbour $u$ of $v$ with $\dist(O,u) = \dist(O,v)-1$.
	By assumption, $u$ is numbered correctly and thus $\lambda(u) = \dist(O,u)\bmod{3} = (\dist(O,v)-1)\bmod{3} \neq (\lambda(v)-1)\bmod{3}$.
	Furthermore, $\lambda(u) = \lambda(v)$ for the neighbouring nodes $u,v$ is excluded, so it must be the case that $\lambda(u) = (\lambda(v)+1)\bmod{3}$.
	According to the initial induction, either an error state occurs, or after transition $\dist(O,v)-1$, $u$ is in stage $1$, while $v$ is still in stage $0$.
	The latter situation also produces an error state next by the second \ref{QL-stage1} transition.
	So, in any case, there will be an error state that propagates through the graph with the second \ref{QL-error} transition, rejecting $G$.
	
	Lastly, we have to deal with the case that all nodes fulfil \hyperref[itm:nqlg_numbering]{(QL1)} and \hyperref[itm:nqlg_neighbours]{(QL2)}, but the length of the graph is not unambiguous, that is, there is no $n \in \mathbb{N} \setminus \{0\}$, such that \hyperref[itm:nqlg_length]{(QL3)} is satisfied.
	Therefore, there are limit nodes $l_0,l_2 \in V$ with $\dist(O,l_0) \neq \dist(O,l_2)$.
	
	We will show that there is always an agent that can detect this fault.
	To do that formally, we have to introduce a few new concepts:
	A directed edge $(v,w)$ is \textit{ascending} if{}f $\lambda(w) = (\lambda(v)+1)\bmod{3}$, and \textit{descending} if{}f $\lambda(w) = (\lambda(v)-1)\bmod{3}$, respectively.
	A directed path is called \textit{ascending} (\textit{descending}) if{}f all edges are ascending (descending).
	A \textit{differing limits-path} $P$ is a directed path between two limit nodes $m_0,m_1 \in V$ with $\dist(O,m_0) \neq \dist(O,m_1)$.
	Since limit nodes do not have successors, $P$ has to start with a descending edge and end with an ascending edge.
	Therefore, there must be a node where $P$ switches from descending to ascending edges.
	A node in $P$ where this happens, is called a \textit{turning node}.
	
	We want to show that, in $G$, there is a differing limits-path with exactly one turning node.
	Let $P = (p_1,...,p_k)$ be a differing limits-path between $l_0,l_2$ with $t$ turning nodes.
	If $t\geq2$, consider the first two turning nodes $u,w$ of $P$.
	Evidently, there must be a node $v$ in $P$ between $u$ and $w$ where $P$ switches from ascending to descending edges again; say the corresponding edges incident to $v$ are $p_r,p_{r+1}$.
	Let $Q = (q_1,...,q_j)$ be an ascending directed path from $v$ to some limit node $l_1$.
	Since $\dist(O,l_0) \neq \dist(O,l_2)$, at least one of the non-equalities $\dist(O,l_1) \neq \dist(O,l_0)$ or $\dist(O,l_1) \neq \dist(O,l_2)$ must be true.
	If the former is true, consider the directed path $(p_1,...,p_r,q_1,...,q_j)$ -- it starts at $l_0$ and ends at $l_1$ with $\dist(O,l_0) \neq \dist(O,l_1)$, so it is still a differing limits-path; it keeps the first turning node $u$ from $P$ between $p_1$ and $p_r$ and adds no further turning nodes as $p_r$ and $q_1$ to $q_j$ are all ascending, for a total of $1<t$ turning nodes.
	Otherwise, consider the directed path $(\accentset{\leftharpoonup}{q_j},...,\accentset{\leftharpoonup}{q_1},p_{r+1},...,p_k)$, where for an edge $e=(v_0,v_1)$, $\accentset{\leftharpoonup}{e}$ denotes the reversed edge $(v_1,v_0)$, in particular making ascending edges descending -- this path starts at $l_1$ and ends at $l_2$ with $\dist(O,l_1) \neq \dist(O,l_2)$, so it is still a differing limits-path; it keeps the second turning node $w$ and all $t-2$ later turning nodes from $P$ between $p_{r+1}$ and $p_k$, but adds no further turning nodes as $\accentset{\leftharpoonup}{q_j}$ to $\accentset{\leftharpoonup}{q_1}$ and $p_{r+1}$ are all descending, for a total of $t-1<t$ turning nodes.
	In conclusion, we constructed a differing limits-path with fewer than $t$ turning nodes in both cases.
	By iterating this construction, we obtain a differing limits-path with exactly one turning node.
	
	Using this preparation, we construct a node that can detect the fault as claimed:
	Let $v \in V$ maximise $\dist(O,v)$ among all nodes that are the only turning node of a differing limits-path.
	Let $P$ be a differing limits-path between limit nodes $l_0,l_1 \in V$ with $v$ as its only turning node, such that $\dist(O,l_0)$ is minimal.
	From $P$, we obtain ascending directed paths $L_0$ from $v$ to $l_0$ with length $k_0$ starting with the edge $(v,w_0)$ and $L_1$ from $v$ to $l_1$ with length $k_1$ starting with the edge $(v,w_1)$.
	Since $\dist(O,l_0)$ is minimal and $\dist(O,l_0) \neq \dist(O,l_1)$, we know
	$$\dist(O,v)+k_0 = \dist(O,l_0) < \dist(O,l_1) = \dist(O,v)+k_1 \iff k_0<k_1.$$
	As $\dist(O,v)$ was chosen maximally, $w_0,w_1$ with $\dist(O,w_0) = \dist(O,v)+1 = \dist(O,w_1)$ cannot be turning nodes of differing limits-paths with exactly one turning node.
	Therefore, all limit nodes $l'$ reachable via an ascending directed path from $w_0$ ($w_1$) must fulfil $\dist(O,l') = \dist(O,l_0) = \dist(O,v)+k_0$ ($\dist(O,l') = \dist(O,l_1) = \dist(O,w_1)+(k_1-1)$).
	
	Now, for the detection of the fault, assuming that no other error states occur:
	According to the initial induction, $v$ reaches stage $1$ in transition $\dist(O,v)$.
	Since the numbering is correct, in the next $k_0$ transitions, stage $1$ states propagate along $L_0$ using \ref{QL-stage1}, so that $l_0$ reaches stage $1$ in transition $\dist(O,v)+k_0$, and then employs \ref{QL-limit} to transition to stage $2$ in transition $\dist(O,v)+k_0+1$.
	In the next $k_0$ transitions after that, stage $2$ states propagate along $L_0$ in reverse direction using \ref{QL-stage2}, so that $v$ reaches stage $2$ in transition $\dist(O,v)+2k_0+1$.
	Analogously, we deduce that $w_1$ reaches stage $2$ in transition $\dist(O,w_1)+2(k_1-1)+1 \geq \dist(O,w_1)+2k_0+1 = \dist(O,v)+2k_0+2$.
	Crucially, the predecessor $v$ of $w_1$ reaches stage $2$ before $w_1$ itself, causing $w_1$ to produce a $\bot$ state by \ref{QL-stage2}.
	This concludes that an error state will always occur and propagate through the graph with the second \ref{QL-error} transition, rejecting $G$.
\end{proof}

\subsubsection{Proof of \Cref{thm:reduction} for $\mathtt{dAf}$ and $\mathtt{dAF}$}
\input{chapters/appendix_theorem_2.tex}

%% file: chapters/appendix_theorem_2.tex
\reduction*

\begin{proof}[Proof for $\mathtt{dAf}$ and $\mathtt{dAF}$]
	Given a TM $T$, let $A^L$ be the $\mathtt{dA\texttt{\$}f}$-distributed automaton from \Cref{lem:nqlg}, and $A^H$ the $\mathsf{NLG}$-distributed automaton obtained from applying \Cref{cor:always_acc/rej} to $T$ and lifting it to the class $\mathtt{dA\texttt{\$}f}$, which is possible by the hierarchy of expressive power.
	For an NQLG $\tilde{G} \in \mathsf{NQLG}$ of length $n \in \mathbb{N} \setminus \{0\}$, \Cref{prop:nqlg_eq_nlg} yields that $A^H$ accepts (rejects) $\tilde{G}$ if{}f $A^H$ accepts (rejects) the NLG $G \in \mathsf{NLG}$ of the same length $n$.
	Therefore, $A^H$ is, in fact, an $\mathsf{NQLG}$-distributed automaton.
	As $L(A^L) = \mathsf{NQLG}$, \Cref{lem:intersection} thus yields a $\mathtt{dA\texttt{\$}f}$-distributed automaton $A^T$ with $L(A^T) = L(A^L) \cap L(A^H)$.
	We show that $L(A^T) \neq \emptyset$ if{}f $T$ halts on blank tape.
	
	If $L(A^T) \neq \emptyset$, let $\tilde{G} \in L(A^T) = L(A^L) \cap L(A^H)$.
	As $\tilde{G} \in L(A^L) = \mathsf{NQLG}$, $\tilde{G}$ is an NQLG of some length $n \in \mathbb{N} \setminus \{0\}$.
	By \Cref{prop:nqlg_eq_nlg}, $A^T$ thus also accepts the NLG $G$ of the same length $n$.
	As $G \in L(A^H)$ for $G$ an NLG, $T$ halts on blank tape by \Cref{cor:always_acc/rej}.
	
	Conversely, if $T$ halts on blank tape, then so does $T_\infty$ from \Cref{lem:always_halt/overflow}.
	This happens after a finite number of steps, visiting a finite number of tape cells $n_0 \in \mathbb{N} \setminus \{0\}$.
	By \Cref{lem:nqlg,cor:always_acc/rej}, $A^L$ and $A^H$ both accept the NLG $G \in \mathsf{NLG}$ of length $n_0$ (note every NLG is an NQLG).
	We conclude $G \in L(A^L) \cap L(A^H) = L(A^T)$, so $L(A^T) \neq \emptyset$.
	
	By \Cref{prop:sync} and the hierarchy of expressive power, this shows the claim for the classes $\mathtt{dAf}$ and $\mathtt{dAF}$.
\end{proof}

%% file: chapters/appendix_snowballfight.tex
\subsection{Proofs about \Cref{pre_cstr:snowball}}\label{app:snowball}

We start by formalizing \Cref{pre_cstr:snowball}:

\begin{construction}[Snowball Fight!\ (formal)]\label{cstr:snowball}
	We construct a $\mathtt{Da\texttt{\$}f}$-distributed machine $M^L = (Q^L,\delta_0^L,\delta^L,Y^L,N^L)$ with:
	\begin{itemize}
		\item The labelling alphabet is $\Lambda^L = \mathbb{Z}_3 \times \{-1,+1\} \times \{0,1\}$. The first component is the numbering, the second component is the direction the agent is facing at the beginning, and the third component indicates whether the agent is holding a snowball at the beginning.
		\item $Q^L = \Lambda \cup (\{\circ\} \times \Lambda) \cup \{\checked, \square, \bot\}.$
		The states are the triples from the labelling alphabet with an optional marker that indicates that the snowball fight has not started yet.
		Furthermore, the checked checkbox $\checked$ and $\bot$ are the accepting and rejecting states, respectively, and the unchecked checkbox $\square$, indicates the intention to accept once it is made sure that no rejecting states have occurred elsewhere.
		\item $\delta_0^L \colon {\Lambda \to Q^L, l \mapsto (\circ,l)}$.
		An agent's initial state is exactly the triple from its label, but with the marker set.
		\item As mentioned, the accepting state is $Y^L = \{\checked\}$, while the rejecting state is $N^L = \{\bot\}$.
	\end{itemize}
	Lastly, we specify the transition function $\delta^L$.
	On one hand, it checks that the numbering is correct; on the other hand, it establishes the rules for the snowball fight.
	For an agent $v \in V$ and a configuration $C \colon {V \to Q}$, we define:
	\begin{equation}
		\begin{aligned}
			\big((\circ,n,d,s),\mathcal{N}_v^C\big) & \mapsto \bot & \text{for } & \mathcal{N}_v^C(\circ,n,*,*) \geq 1\ \vee \\
			&&& \exists m \in \mathbb{Z}_3 \colon {\mathcal{N}_v^C(\circ,m,*,*) \geq 2} \\
			\big(q,\mathcal{N}_v^C\big) & \mapsto \bot & \text{for } & q \neq \checked \wedge \mathcal{N}_v^C(\bot) \geq 1
		\end{aligned}\labelledtag{SF-error}
	\end{equation}
	Faults in the numbering are detected by the first transition exactly as by \ref{L-error} in \Cref{cstr:nlg}.
	Error propagation also works the same way, with the only difference that accepting states cannot be overwritten to comply with halting acceptance ($\mathtt{a}$).
	For the following transitions, we assume that, in addition to the specified conditions being fulfilled, the conditions for the \ref{SF-error} transitions are not met and that $\mathcal{N}_v^C(\checked) = \mathcal{N}_v^C(\square) = 0$:
	\begin{equation}
		\begin{aligned}
			\big((\circ,n,d,s),\mathcal{N}_v^C\big) & \mapsto (n,d,s) && \text{for } (s = 1) \oplus (\mathcal{N}_v^C(\circ,*,*,1) \geq 2) \\
			\big((\circ,n,d,s),\mathcal{N}_v^C\big) & \mapsto \bot && \text{otherwise}
		\end{aligned}\labelledtag{SF-init}
	\end{equation}
	Before removing the marker and thereby starting the snowball fight, every agent checks that exactly every other agent, i.e., either they themselves or both of their neighbours, is holding a snowball.
	The fact that only every other agent can start with a snowball, establishes that two neighbouring agents are not holding a snowball at the beginning of the fight, which will turn out to be an important invariant throughout the whole execution of the machine.
	The fact that exactly every other agent has to start with a snowball, makes it easy for every agent to check that, at the beginning, snowballs even exist and that the invariant is fulfilled.
	If this is not fulfilled, the agent goes into the rejecting state, as a correct execution of the snowball fight cannot be guaranteed.
	\begin{align*}
		\big((n,d,1),\mathcal{N}_v^C\big) & \mapsto (n,d,0) & \text{for } & \mathcal{N}_v^C((n+d)\bmod{3},*,0) \geq 1 \labelledtag{SF-throw} \\
		\big((n,d,0),\mathcal{N}_v^C\big) & \mapsto (n,-d,1) & \text{for } & \mathcal{N}_v^C((n+d)\bmod{3},-d,1) \geq 1 \labelledtag{SF-catch} \\
		\big((n,d,0),\mathcal{N}_v^C\big) & \mapsto \bot & \text{for } & \mathcal{N}_v^C((n+d)\bmod{3},-d,1) = 0\ \wedge \\
		&&& \mathcal{N}_v^C((n-d)\bmod{3},d,1) \geq 1 \labelledtag{SF-hit}
	\end{align*}
	If an agent is holding a snowball, it will throw it at the agent it is facing (which will definitely not be holding a snowball due to the invariant mentioned above).
	The receiving agent reacts in one of three ways, also illustrated in \Cref{fig:snowball_tra}:
	If it is facing towards the thrower, it catches the snowball and turns around by \ref{SF-catch}.
	If it is facing away from the thrower, but facing another thrower, it catches and merges both snowballs thrown at it, while turning around, again by \ref{SF-catch} -- the overall number of snowballs decreased by one.
	Lastly, if it is facing away from the thrower and is not facing another thrower, it gets hit and produces an error state by \ref{SF-hit}.
	
	The only transitions missing are those enabling acceptance, which obviously do allow for $\mathcal{N}_v^C(\checked), \mathcal{N}_v^C(\square) \neq 0$ again:
	\begin{equation}
		\begin{aligned}
			\big((0,-1,1),\mathcal{N}_v^C\big) & \mapsto \square & \text{for } & \mathcal{N}_v^C(2,*,*) = 0 \\
			\big((n,d,1),\mathcal{N}_v^C\big) & \mapsto \bot & \text{for } & \mathcal{N}_v^C((n+d)\bmod{3},*,*) = 0 \wedge (n,d) \neq (0,-1) \\
			\big((n,d,s),\mathcal{N}_v^C\big) & \mapsto \square & \text{for } & \mathcal{N}_v^C(\square) \geq 1 \wedge \mathcal{N}_v^C((n+1)\bmod{3},*,*) \geq 1 \\
			\big((n,d,s),\mathcal{N}_v^C\big) & \mapsto \checked & \text{for } & \mathcal{N}_v^C(\square) \geq 1 \wedge \mathcal{N}_v^C((n+1)\bmod{3},*,*) = 0 \\
			\big(\square,\mathcal{N}_v^C\big) & \mapsto \checked & \text{for } & \mathcal{N}_v^C(\checked) \geq 1 \vee \mathcal{N}_v^C(*) = 0
		\end{aligned}\labelledtag{SF-end}
	\end{equation}
	Recall that the key difficulty for automata with halting acceptance ($\mathtt{a}$) was ensuring the existence of the origin node.
	If the agent at the origin node is holding a snowball and is facing away from its right neighbour, the origin node obviously exists, enabling it to transition to the $\square$ state to indicate its intention to accept.
	Before accepting however, it has to be ensured that no error state has occurred elsewhere in the graph, which would make acceptance impossible due to halting acceptance ($\mathtt{a}$).
	So, the $\square$ state propagates through the graph, overwriting every 3-component state, but in particular not the rejecting $\bot$ states.
	If a $\bot$ state did actually occur, its unconditional propagation according to \ref{SF-error} would outweigh the $\square$-propagation, rejecting the graph.
	If the $\square$ states reach the last node, it means that there were no errors in the graph and the graph can be accepted.
	The $\square$ states are replaced by accepting $\checked$ states one by one.
	The condition $\mathcal{N}_v^C(*) = 0$, can obviously only be fulfilled if the input graph only has a single node, dealing with that edge case.
	If an agent outside the origin node that does not have a left (right) neighbour is holding a snowball and is facing away from its right (left) neighbour, an error occurs as the existence of the origin node cannot be guaranteed.
	
	All other transitions are silent.
\end{construction}

We proceed to prove the three lemmas teased in \Cref{sec:snowball}.
Observe that an input graph $G$ with only a single node is accepted if{}f that node has labelling $(0,-1,1)$, according to the \ref{SF-init} and \ref{SF-end} transitions.
As this fully characterizes the behaviour of $M^L$ on single node-graphs and aligns with the following two lemmas, we will from now on assume that $G$ has at least two nodes.
In particular, every node has at one least neighbour.

\begin{lemma}\label{lem:sf_error}
	Let $M^L$ be \Cref{cstr:snowball} and $G = (V,E,\lambda)$ be a labelled graph.
	Further, let $\rho$ be the (unique) run of $M^L$ on $G$.
	Assume that a $\bot$ state occurs in $\rho$.
	Then $\rho$ is rejecting.
\end{lemma}
\begin{proof}
	If no accepting state occurs, then $\rho$ is obviously rejecting, as the $\bot$ state will just propagate through the graph by the second \ref{SF-error} transition.
	It therefore suffices to show that accepting states cannot occur in $\rho$.
	
	For this purpose, we analyse the requirements for an accepting $\checked$ state to occur:
	Looking at the \ref{SF-end} transitions, we see that for the first accepting state to occur, there has to be an origin node to produce the first $\square$ state.
	This $\square$ state then has to propagate through the graph, until it reaches a node without a successor to produce a $\checked$ state.
	
	Let $w$ be a node without a successor.
	All of $w$'s neighbours are either predecessors or have the same labelling as $w$ itself.
	If $w$ has a neighbour with the same labelling as itself or $w$ has multiple predecessors, it produces a $\bot$ state according to the first \ref{SF-error} transition and can therefore not produce a $\checked$ state.
	It remains the case where $w$ has exactly one neighbour, and this neighbour is a predecessor.
	Since $w$ is not an origin node, the $\square$ state required to produce a $\checked$ state at $w$ would have to come from a different agent $u \neq w$ which is an origin node.
	For the origin node $u$ to produce a $\square$ state by the first \ref{SF-end} transition, rather than a $\bot$ state by the first \ref{SF-error} transition, it has to have degree $1$ as well.
	Let $P = (\{u,v_0\},...,\{v_k,w\})$ be a simple path from $u$ to $w$.
	If any node $v_0,...,v_k$ has degree greater than $2$, an error state will be produced by the first \ref{SF-error} transition.
	The error state will hinder the $\square$ state produced at $u$, from reaching $w$ through $P$.
	If there is a different simple path $Q \neq P$ from some origin node to $w$, it has to join $P$ at some node $v_i,i \in [k]$ since $w$ only has one neighbour $v_k$.
	Therefore, $v_i$ is a node with degree greater than $2$ on $Q$, implying that the $\square$ state cannot propagate to $w$ through $Q$ either.
	Thus, if a $\square$ state were to propagate from any origin node to $w$, $G$ would have to be a linear graph with $E = \{u,v_0,...,v_k,w\}$.
	Clearly, after an agent has reached the $\square$ state, it cannot be the first to produce an error state; it could only adopt a propagating error state by the second \ref{SF-error} transition.
	So, the first error state in $\rho$ has to occur at an agent before it transitions to the $\square$ state, thus blocking the $\square$ state from reaching $w$.
	Therefore, no accepting states can occur in $\rho$, completing the proof.
\end{proof}

\begin{lemma}\label{lem:sf_nerror}
	Let $M^L$ be \Cref{cstr:snowball} and $G = (V,E,\lambda)$ be a labelled graph.
	Further, let $\rho$ be the (unique) run of $M^L$ on $G$.
	Assume that no $\bot$ state occurs in $\rho$.
	Then $G$ is a snowball fight numbered linear graph and $\rho$ is accepting.
\end{lemma}
\begin{proof}
	The proof will comprise multiple steps:
	First, we notice that $G$ has to be an SFNLG or SFNCG.
	Then, we establish a few facts about possible movement patterns of snowballs, from which we will be able to conclude that the total number of snowballs has to eventually become one.
	Finally, we prove that the last snowball will eventually disappear as well and can only do so without producing an error at the origin node, ultimately causing acceptance.
	This also implies that $G$ cannot actually be an SFNCG as it has to have an origin node.
	
	We start by observing, that every node has to have degree at most $2$ and that the numbering of the graph is correct.
	Otherwise, the first \ref{SF-error} transition would produce a $\bot$ state, analogous to \ref{L-error} in the proof of \Cref{lem:nlg}.
	$G$ therefore is an SFNLG or SFNCG.
	
	As there are no errors after removing the markers by \ref{SF-init}, we know that exactly every other agent is holding a snowball at the start.
	This enables us to define a bipartition of $V$ with parts $U_i = \{v \in V \mid \lambda(v)_3 = i\}$ for $i \in \{0,1\}$.
	With this, we can formalize the invariant mentioned in the construction of $M^L$:
	In any configuration of $\rho$, either only agents in $U_0$ or only agents in $U_1$ can be holding a snowball.
	In the beginning and after the first transition, exactly the agents in $U_1$ are holding a snowball.
	Assume that only agents $v \in U_i,i \in \{0,1\}$ are holding a snowball.
	By construction of the bipartition, $v$'s neighbours are in $U_{1-i}$ and are surely not holding a snowball.
	So, if $v$ is holding a snowball, it will next perform \ref{SF-throw} or one of the \ref{SF-end} transitions.
	If $v$ is not holding a snowball, it will definitely not perform \ref{SF-catch} next, which would be the only way for $v$ to obtain a snowball.
	In any case, $v$ will not be holding a snowball after the next transition.
	As $v \in U_i$ was arbitrary, only agents in $U_{1-i}$ can be holding a snowball after the next transition.
	Inductively, this proves the invariant.
	
	Observe that \ref{SF-catch} is the only transition where a snowball appears, i.e., an agent $v \in V$ that was not holding a snowball before the transition is afterwards.
	However, the condition for \ref{SF-catch} requires that the neighbour $w$ that $v$ is facing, has to be holding a snowball and facing $v$ before the transition.
	Therefore, if $v$ performs \ref{SF-catch}, $w$ has to perform \ref{SF-throw} simultaneously.
	Thus, in the same transition where a snowball appears at $v$, a snowball has to disappear at $w$ -- intuitively, the snowball is passed.
	Furthermore, if snowballs appear at two different agents $v_0 \neq v_1$, snowballs also have to disappear at two different neighbours $w_0 \neq w_1$, since each $w_i,i \in \{0,1\}$ can only be facing one of the $v_i,i \in \{0,1\}$.
	We can conclude that the total number of snowballs in $G$ can never increase.
	
	Consider a configuration with at least two snowballs and no $\square$ states.
	For an agent $u \in V$ facing in direction $d_u$ and holding a snowball, we define $k(u)$ as the length of the chain of consecutive agents located in direction $d_u$ of $u$ that are facing in direction $-d_u$.
	Intuitively, $k(u)$ describes how many times the snowball could be passed in direction $d_u$.
	Let $k_0$ be the minimum $k(u)$ among all $u \in V$ holding a snowball.
	If $k_0 \geq 1$, all $u \in V$ holding a snowball will perform a pass to their neighbour in direction $d_u$, shortening the chain of consecutive agents facing in direction $-d_u$ by one on that end.
	In particular, no $\square$ states are produced.
	If the chain does not simultaneously extend on the other end, its length decreases by one.
	Assuming the chain does extend, we know that the agent $w_u$ after the last agent in the chain must have flipped the direction it is facing from $d_u$ to $-d_u$.
	This can only happen by \ref{SF-catch}, so $w_u$ must then be holding a snowball.
	$w_u$'s neighbours cannot be holding snowballs because of the invariant, so the directions they are facing remained unchanged.
	The neighbour of $w_u$ in direction $-d_u$ is therefore facing in direction $-d_u$ as before, implying $k(w_u) = 0$ and thus $k_0$ decreases to $0$.
	Since the number of snowballs in $G$ can never increase, it is sufficient to only consider what happens to existing chains.
	In conclusion, without producing $\square$ states, either all chains decrease their lengths by one, resulting in $k_0$ decreasing by one, or at least one chain extends implying that $k_0$ decreases to $0$.
	In any case, we eventually get $k_0 = 0$, meaning that there is at least one agent $u_0$ facing in direction $d_0$, but there is no neighbour $v_0$ in direction $d_0$ that is facing in direction $-d_0$.
	If $v_0$ exists, it must be facing in direction $d_0$.
	The only way for $v_0$ not to produce an error state by \ref{SF-hit} is if $v_0$'s neighbour $w_0$ in direction $d_0$ is facing in direction $-d_0$ and also holding a snowball.
	In the next transition, the snowballs at $u_0$ and $w_0$ have to disappear by \ref{SF-throw} and only one snowball appears at $v_0$ by \ref{SF-catch} -- the total number of snowballs decreases, while at least one snowball remains.
	If the neighbour $v_0$ does not exist, $u_0$ must be an origin node and perform the first \ref{SF-end} transition to produce a $\square$ state rather than a $\bot$ state.
	We conclude that eventually the total number of snowballs in $G$ decreases to one or a $\square$ state is produced at an origin node.
	
	In the case where no $\square$ state is produced before the number of snowballs becomes exactly one, let $u_1 \in V$ be the agent holding the last snowball and facing in direction $d_1$.
	If $u_1$ is an origin node, it performs the first \ref{SF-end} transition next, producing a $\square$ state.
	Otherwise, $u_1$ has to pass the snowball to its neighbour $v_1$ in direction $d_1$, whereby $v_1$ turns around from facing in direction $-d_1$ to $d_1$ -- the total number $k_1$ of agents facing in direction $-d_1$ decreases by one.
	This implies that eventually $k_1 = 0$ or a $\square$ state is produced at an origin node.
	If no $\square$ state is produced before $k_1 = 0$, the agent holding the last snowball has to be an origin node and produce a $\square$ state next, otherwise an error would occur as there are no agents left facing in direction $-d_1$.
	Ultimately, we get that eventually an origin node has to produce a $\square$ state.
	Since an SFNCG does not have an origin node, this rules out the possibility that $G$ could have been an SFNCG -- we have $G \in \mathsf{SFNLG}$.
	
	Lastly, since no error states occur, the \ref{SF-end} transitions make the run halt as follows:
	Once an origin node has produced a $\square$ state, it propagates through the SFNLG $G$ until it reaches the last node, which does not have a right neighbour.
	This node then produces an accepting $\checked$ state, which propagates back through the graph, accepting $G$.
\end{proof}

Lastly, we construct an infinite (non-exhaustive) family $\{L_n\}_{n \in \mathbb{N} \setminus \{0\}}$ of SFNLGs that get accepted by \Cref{cstr:snowball}.
While the formal definition takes a slightly different approach from the sketch presented in \Cref{sec:snowball}, the resulting family remains identical.

\begin{definition2}[Harmonious Snowball Fight Numbered Linear Graphs]\label{def:snowball_harmonious}
	We define two families of words $\{l_n\}_{n \in \mathbb{N} \setminus \{0\}}$ and $\{r_n\}_{n \in \mathbb{N} \setminus \{0\}}$ over $\{-1,+1\} \times \{0,1\}$ inductively by
	\begin{bracketenumerate}
		\item $l_1 = (-1,1)$ and $r_1 = (+1,1)$,
		\item $l_{n+1} = r_n\ (+1,0)\ l_n$ and $r_{n+1} = r_n\ (-1,0)\ l_n$ for all $n \in \mathbb{N} \setminus \{0\}$.
	\end{bracketenumerate}
	We define a family of harmonious snowball fight graphs $\{L_n = (V_n, E_n, \lambda_n)\}_{n \in \mathbb{N} \setminus \{0\}}$ with
	\begin{bracketenumerate}
		\item $V_n = \{v_0,...,v_{|l_n|-1}\}, E_n = \{\{v_i,v_{i+1}\} \mid i \in [|l_n|-2]\}$ ($L_n$ is linear), and
		\item $\forall i \in [|l_n|-1] \colon {\lambda_n(v_i) = (i\bmod{3},{(l_n)}_i)}$ (agent numbering modulo 3, second (direction) and third (snowball) component according to $l_n$).
	\end{bracketenumerate}
\end{definition2}

\begin{lemma}\label{lem:sf_existence}
	Let $M^L$ be \Cref{cstr:snowball} and $\{L_n\}_{n \in \mathbb{N} \setminus \{0\}}$ as defined in \Cref{def:snowball_harmonious}.
	For all $n \in \mathbb{N} \setminus \{0\}$, the (unique) run $\rho_n$ of $M^L$ on $L_n$ is accepting.
\end{lemma}
\begin{proof}
	We show that in harmonious snowball fight graphs, the last snowball eventually leaves the graph to the left causing acceptance, as sketched in \Cref{sec:snowball}.
	To prove this, we define the family of graphs $\{R_n\}_{n \in \mathbb{N} \setminus \{0\}}$ analogously to harmonious snowball fight graphs $\{L_n\}_{n \in \mathbb{N} \setminus \{0\}}$, but with the second and third component of the labelling according to $r_n$ rather than $l_n$.
	This corresponds to the mirrored graphs from the sketch.
	Note that the first \ref{SF-error} transition is not performed in $L_n$ ($R_n$), as the numbering is correct.
	Further note that for all $n \in \mathbb{N} \setminus \{0\}$, it follows readily from the inductive construction of $l_n$ and $r_n$ that $|l_n| = |r_n| \eqcolon k_n$, and that exactly every other agent in $L_n$ ($R_n$) starts with a snowball.
	Transition $0$ (we number the transitions starting with $0$) on $L_n$ ($R_n$) thus only performs the first \ref{SF-init} transition.
	
	As in the sketch, we now formally show for all $n \in \mathbb{N} \setminus \{0\}$:
	\begin{bracketenumerate}
		\item Transitions $1$ to $k_n-1$ only include \ref{SF-throw} and \ref{SF-catch} transitions, and\label{itm:sfe_1}
		\item after transition $k_n-1$, all agents of $L_n$ ($R_n$) are facing left (right) and the leftmost (rightmost) agent is holding the only snowball.\label{itm:sfe_2}
	\end{bracketenumerate}
	\Cref{fig:snowball_harmonious} illustrates exemplarily that this is true for $L_3$.
	For the general proof we proceed by induction on $n$.
	
	For the base case $n = 1$, we consider $L_1$ ($R_1$).
	Since $k_1=1$, there is only claim \ref{itm:sfe_2} left to show, which is fulfilled as the only snowball is held by the only and thereby leftmost (rightmost) agent, which is facing left (right).
	
	For the inductive step, we assume that claims \ref{itm:sfe_1} and \ref{itm:sfe_2} hold for some $n \in \mathbb{N} \setminus \{0\}$ and show them for $n+1$.
	From the inductive definition of $l_{n+1}$ and $r_{n+1}$, we see that the second (direction) and third (snowball) component of the labelling of $L_{n+1}$ ($R_{n+1}$) are exactly the same as in a copy of $R_n$ with a single node $v_{k_n}$ attached to the rightmost node and a copy of $L_n$ attached to $v_{k_n}$, where $v_{k_n}$ is facing right (left) and not holding a snowball.
	Assume that, in transitions $1$ to $k_n-1$, $v_{k_n}$ performs a non-silent transition or affects a transition performed by one of its neighbours $w \in \{v_{k_n-1},v_{k_n+1}\}$.
	This would imply that $w$ would have performed the first, second, or fourth \ref{SF-end} transition if it were not for $v_{k_n}$.
	However, by claim \ref{itm:sfe_1} of the induction hypothesis, $v_{k_n-1}$ and $v_{k_n+1}$ (as part of the subgraph of $L_{n+1}$ ($R_{n+1}$) corresponding to $R_n$ and $L_n$, respectively) only perform \ref{SF-throw} and \ref{SF-catch} in transitions $1$ to $k_n-1$.
	Thus, after transition $k_n-1$, $v_{k_n}$ has to still be facing right (left) and not be holding a snowball.
	Furthermore, claim \ref{itm:sfe_2} of the induction hypothesis yields that after transition $k_n-1$, the nodes $v_0,...,v_{k_n-1}$ are facing right, the nodes $v_{k_n+1},...,v_{2k_n+1}$ are facing left and exactly the nodes $v_{k_n-1}$ and $v_{k_n+1}$ are holding a snowball each.
	Applying \ref{SF-throw} to $v_{k_n-1}$ and $v_{k_n+1}$ and \ref{SF-catch} to $v_{k_n}$, merges the two snowballs into one, held by $v_{k_n}$ while facing left (right).
	Since the nodes $v_0,...,v_{k_n-1}$ ($v_{k_n+1},...,v_{2k_n+1}$) are all facing right (left), the next $k_n$ transitions each apply a \ref{SF-throw} and a \ref{SF-catch} to a pair of neighbouring nodes, so that the snowball is passed one node to the left (right), while flipping the direction the receiving node is facing from right (left) to left (right).
	In conclusion, transitions $1$ to $(k_n-1)+1+k_n=|r_n|+1+|l_n|-1=k_{n+1}-1$ only include \ref{SF-throw} and \ref{SF-catch} transitions, and after transition $k_{n+1}-1$ all agents are facing left (right) with the leftmost (rightmost) agent holding the only snowball.
	This completes the induction.
	
	Finally, we can prove the claim of the lemma.
	Consider $L_n$ for $n \in \mathbb{N} \setminus \{0\}$.
	From the induction above, specifically claim \ref{itm:sfe_2}, we know that after transition $k_n-1$, all agents are facing left and the leftmost agent $v_0$ is holding the only snowball.
	$v_0$ is also the origin node, so it can perform the first \ref{SF-end} transition, producing a $\square$ state and consuming the last snowball.
	So far no error states occurred, and as there are no snowballs left, no error states will occur any more.
	Thus, \Cref{lem:sf_nerror} yields that $\rho_n$ is accepting.
\end{proof}

\subsection{Proof of \Cref{thm:reduction} for $\mathtt{Daf}$ and $\mathtt{DaF}$}
\input{chapters/appendix_theorem_3.tex}

%% file: chapters/appendix_theorem_3.tex
\reduction*

\begin{proof}[Proof for $\mathtt{Daf}$ and $\mathtt{DaF}$]
	Given a TM $T$, let $A^L$ be the $\mathtt{Da\texttt{\$}f}$-distributed automaton from \Cref{cstr:snowball}, and $A^H = ((Q^H,\delta_0^H,\delta^H,Y^H,N^H), \Sigma)$ the $\mathsf{NLG}$-distributed automaton obtained from applying \Cref{cor:always_acc/rej} to $T$ and lifting it to the class $\mathtt{Da\texttt{\$}f}$, which is possible by the hierarchy of expressive power.

	Recall that $A^L$ operates on the labelling alphabet $\Lambda^L = \mathbb{Z}_3 \times \{-1,+1\} \times \{0,1\}$, while $A^H$ operates on $\Lambda^H = \mathbb{Z}_3$.
	To enable $A^H$ to operate on $\Lambda^L$ as well, we replace $\delta_0^H$ by $\delta_0^H \circ \pi_1$ where $\pi_1 \colon {\Lambda^L \to \Lambda^H, (n,d,s) \mapsto n}$ is the projection onto the first component.
	As $A^H$ is an $\mathsf{NLG}$-distributed automaton, $A^{\mathit{SFH}} \coloneq ((Q^H,\delta_0^H \circ \pi_1,\delta^H,Y^H,N^H),\Sigma)$ is an $\mathsf{SFNLG}$-distributed automaton.
	As $L(A^L) \subseteq \mathsf{SFNLG}$ by \Cref{lem:sf_error,lem:sf_nerror}, $A^{\mathit{SFH}}$ is, in particular, an $L(A^L)$-distributed automaton.
	\Cref{lem:intersection} thus yields a $\mathtt{Da\texttt{\$}f}$-distributed automaton $A^T$ with $L(A^T) = L(A^L) \cap L(A^{\mathit{SFH}})$.
	We show that $L(A^T) \neq \emptyset$ if{}f $T$ halts on blank tape.
	
	If $L(A^T) \neq \emptyset$, let $G \in L(A^T) = L(A^L) \cap L(A^{\mathit{SFH}})$.
	As $G \in L(A^L) \subseteq \mathsf{SFNLG}$, $G$ is an SFNLG.
	So, we have $G = (V,E,\lambda) \in L(A^{\mathit{SFH}})$ for $G$ an SFNLG and thus $\pi_1(G) \coloneq (V,E,\pi_1 \circ \lambda) \in L(A^H)$ for $\pi_1(G)$ an NLG.
	This implies that $T$ halts on blank tape by \Cref{cor:always_acc/rej}.
	
	Conversely, if $T$ halts on blank tape, then so does $T_\infty$ from \Cref{lem:always_halt/overflow}.
	This happens after a finite number of steps, visiting a finite number of tape cells $n_0 \in \mathbb{N} \setminus \{0\}$.
	By \Cref{lem:sf_existence}, there is an $n \geq n_0$, such that $A^L$ accepts an SFNLG $G \in \mathsf{SFNLG}$ of length $n$.
	By \Cref{cor:always_acc/rej}, $A^{\mathit{SFH}}$ also accepts $G$.
	We conclude $G \in L(A^L) \cap L(A^{\mathit{SFH}}) = L(A^T)$, so $L(A^T) \neq \emptyset$.
	
	By \Cref{prop:sync} and the hierarchy of expressive power, this shows the claim for the classes $\mathtt{Daf}$ and $\mathtt{DaF}$.
\end{proof}

%% file: SAND_submission.bbl
\begin{thebibliography}{10}

\bibitem{beep_1}
Yehuda Afek, Noga Alon, Ziv Bar{-}Joseph, Alejandro Cornejo, Bernhard Haeupler,
  and Fabian Kuhn.
\newblock Beeping a maximal independent set.
\newblock {\em Distributed Comput.}, 26(4):195--208, 2013.
\newblock URL: \url{https://doi.org/10.1007/s00446-012-0175-7}, \href
  {https://doi.org/10.1007/S00446-012-0175-7}
  {\path{doi:10.1007/S00446-012-0175-7}}.

\bibitem{ni}
Dana Angluin.
\newblock Local and global properties in networks of processors (extended
  abstract).
\newblock In Raymond~E. Miller, Seymour Ginsburg, Walter~A. Burkhard, and
  Richard~J. Lipton, editors, {\em Proceedings of the 12th Annual {ACM}
  Symposium on Theory of Computing, April 28-30, 1980, Los Angeles, California,
  {USA}}, pages 82--93. {ACM}, 1980.
\newblock \href {https://doi.org/10.1145/800141.804655}
  {\path{doi:10.1145/800141.804655}}.

\bibitem{pp_1}
Dana Angluin, James Aspnes, Melody Chan, Michael~J. Fischer, Hong Jiang, and
  Ren{\'{e}} Peralta.
\newblock Stably computable properties of network graphs.
\newblock In Viktor~K. Prasanna, S.~Sitharama Iyengar, Paul~G. Spirakis, and
  Matt Welsh, editors, {\em Distributed Computing in Sensor Systems, First
  {IEEE} International Conference, {DCOSS} 2005, Marina del Rey, CA, USA, June
  30 - July 1, 2005, Proceedings}, volume 3560 of {\em Lecture Notes in
  Computer Science}, pages 63--74. Springer, 2005.
\newblock \href {https://doi.org/10.1007/11502593\_8}
  {\path{doi:10.1007/11502593\_8}}.

\bibitem{pp_2}
Dana Angluin, James Aspnes, Zo{\"{e}} Diamadi, Michael~J. Fischer, and
  Ren{\'{e}} Peralta.
\newblock Computation in networks of passively mobile finite-state sensors.
\newblock {\em Distributed Comput.}, 18(4):235--253, 2006.
\newblock URL: \url{https://doi.org/10.1007/s00446-005-0138-3}, \href
  {https://doi.org/10.1007/S00446-005-0138-3}
  {\path{doi:10.1007/S00446-005-0138-3}}.

\bibitem{beep_2}
Alejandro Cornejo and Fabian Kuhn.
\newblock Deploying wireless networks with beeps.
\newblock In Nancy~A. Lynch and Alexander~A. Shvartsman, editors, {\em
  Distributed Computing, 24th International Symposium, {DISC} 2010, Cambridge,
  MA, USA, September 13-15, 2010. Proceedings}, volume 6343 of {\em Lecture
  Notes in Computer Science}, pages 148--162. Springer, 2010.
\newblock \href {https://doi.org/10.1007/978-3-642-15763-9\_15}
  {\path{doi:10.1007/978-3-642-15763-9\_15}}.

\bibitem{decision_power}
Philipp Czerner, Roland Guttenberg, Martin Helfrich, and Javier Esparza.
\newblock Decision power of weak asynchronous models of distributed computing.
\newblock In Avery Miller, Keren Censor{-}Hillel, and Janne~H. Korhonen,
  editors, {\em {PODC} '21: {ACM} Symposium on Principles of Distributed
  Computing, Virtual Event, Italy, July 26-30, 2021}, pages 115--125. {ACM},
  2021.
\newblock \href {https://doi.org/10.1145/3465084.3467918}
  {\path{doi:10.1145/3465084.3467918}}.

\bibitem{nfsm}
Yuval Emek and Roger Wattenhofer.
\newblock Stone age distributed computing.
\newblock In Panagiota Fatourou and Gadi Taubenfeld, editors, {\em {ACM}
  Symposium on Principles of Distributed Computing, {PODC} '13, Montreal, QC,
  Canada, July 22-24, 2013}, pages 137--146. {ACM}, 2013.
\newblock \href {https://doi.org/10.1145/2484239.2484244}
  {\path{doi:10.1145/2484239.2484244}}.

\bibitem{classification}
Javier Esparza and Fabian Reiter.
\newblock A classification of weak asynchronous models of distributed
  computing.
\newblock In Igor Konnov and Laura Kov{\'{a}}cs, editors, {\em 31st
  International Conference on Concurrency Theory, {CONCUR} 2020, September 1-4,
  2020, Vienna, Austria (Virtual Conference)}, volume 171 of {\em LIPIcs},
  pages 10:1--10:16. Schloss Dagstuhl - Leibniz-Zentrum f{\"{u}}r Informatik,
  2020.
\newblock URL: \url{https://doi.org/10.4230/LIPIcs.CONCUR.2020.10}, \href
  {https://doi.org/10.4230/LIPICS.CONCUR.2020.10}
  {\path{doi:10.4230/LIPICS.CONCUR.2020.10}}.

\bibitem{survey_1}
Ofer Feinerman and Amos Korman.
\newblock Theoretical distributed computing meets biology: {A} review.
\newblock In Chittaranjan Hota and Pradip~K. Srimani, editors, {\em Distributed
  Computing and Internet Technology, 9th International Conference, {ICDCIT}
  2013, Bhubaneswar, India, February 5-8, 2013. Proceedings}, volume 7753 of
  {\em Lecture Notes in Computer Science}, pages 1--18. Springer, 2013.
\newblock \href {https://doi.org/10.1007/978-3-642-36071-8\_1}
  {\path{doi:10.1007/978-3-642-36071-8\_1}}.

\bibitem{wmdc}
Lauri Hella, Matti J{\"{a}}rvisalo, Antti Kuusisto, Juhana Laurinharju, Tuomo
  Lempi{\"{a}}inen, Kerkko Luosto, Jukka Suomela, and Jonni Virtema.
\newblock Weak models of distributed computing, with connections to modal
  logic.
\newblock {\em Distributed Comput.}, 28(1):31--53, 2015.
\newblock URL: \url{https://doi.org/10.1007/s00446-013-0202-3}, \href
  {https://doi.org/10.1007/S00446-013-0202-3}
  {\path{doi:10.1007/S00446-013-0202-3}}.

\bibitem{KuusistoR20}
Antti Kuusisto and Fabian Reiter.
\newblock Emptiness problems for distributed automata.
\newblock {\em Inf. Comput.}, 272:104503, 2020.
\newblock URL: \url{https://doi.org/10.1016/j.ic.2019.104503}, \href
  {https://doi.org/10.1016/J.IC.2019.104503}
  {\path{doi:10.1016/J.IC.2019.104503}}.

\bibitem{survey_2}
Saket Navlakha and Ziv Bar{-}Joseph.
\newblock Distributed information processing in biological and computational
  systems.
\newblock {\em Commun. {ACM}}, 58(1):94--102, 2015.
\newblock \href {https://doi.org/10.1145/2678280} {\path{doi:10.1145/2678280}}.

\bibitem{crn}
David Soloveichik, Matthew Cook, Erik Winfree, and Jehoshua Bruck.
\newblock Computation with finite stochastic chemical reaction networks.
\newblock {\em Nat. Comput.}, 7(4):615--633, 2008.
\newblock URL: \url{https://doi.org/10.1007/s11047-008-9067-y}, \href
  {https://doi.org/10.1007/S11047-008-9067-Y}
  {\path{doi:10.1007/S11047-008-9067-Y}}.

\end{thebibliography}
